\documentclass[a4paper,12pt]{amsart} 

\usepackage[lmargin=2.4cm,rmargin=2.25cm,bmargin=3cm,tmargin=2.5cm]{geometry}

\usepackage[utf8]{inputenc}
\usepackage{amsmath}
\usepackage{amssymb}  
\usepackage{textcomp}
\usepackage{enumerate}
\usepackage[round,sort]{natbib}
\usepackage{graphicx}
\usepackage{booktabs}
\usepackage{diagbox}
\usepackage{multirow}
\usepackage{slashbox}
\usepackage{float}
\usepackage{url}

\usepackage{xcolor}

\newcommand{\defeq}{=}

\newcommand{\B}{\mathsf{B}}

\newcommand{\X}{\mathsf{X}}
\newcommand{\T}{\mathsf{T}}

\newcommand{\mat}[1]{\mathrm{#1}} 

\newcommand{\uarg}{\,\cdot\,}
\newcommand{\ud}{\mathrm{d}}
\providecommand{\abs}[1]{\lvert#1\rvert}
\providecommand{\norm}[1]{\lVert#1\rVert}

\newcommand{\R}{\mathbb{R}}

\renewcommand{\P}{\mathbb{P}}
\newcommand{\E}{\mathbb{E}}
\newcommand{\charfun}[1]{1\left(#1\right)}

\newcommand{\given}{\,:\,}

\newcommand{\var}{\mathrm{var}}
\newcommand{\cov}{\mathrm{cov}}

\newcommand{\bigmid}{\;\big|\;}
\newcommand{\transpose}{{\mathrm{\scriptscriptstyle T}}}

\newcommand{\term}{\mathcal{T}}

\newcommand{\btheta}{\breve{\theta}}
\newcommand{\bTheta}{\breve{\Theta}}
\newcommand{\bx}{\breve{x}}
\newcommand{\by}{\breve{y}}
\newcommand{\bz}{\breve{z}}
\newcommand{\bvartheta}{\breve{\vartheta}}
\newcommand{\pr}{\mathrm{pr}}
\newcommand{\genKernel}{L}

\newcommand{\abcmcmc}{\textsc{abc}-\textsc{mcmc}}
\newcommand{\simulationepsilon}{\delta}

\newcommand{\mref}[1]{\ref{#1}}
\newcommand{\meqref}[1]{(\ref{#1})}

\newtheorem{theorem}{Theorem}

\newtheorem{lemma}[theorem]{Lemma}

\theoremstyle{definition}
\newtheorem{definition}[theorem]{Definition}
\newtheorem{assumption}[theorem]{Assumption}

\newtheorem{algo}[theorem]{Algorithm}

\theoremstyle{remark}

\begin{document}

\title[On the use of ABC-MCMC with inflated tolerance and
post-correction]{On the use of approximate Bayesian computation Markov chain
  Monte Carlo with inflated tolerance and post-correction}
\author{Matti Vihola}
\author{Jordan Franks}
\address{Department of Mathematics and Statistics, University of
Jyväskylä P.O.Box 35, FI-40014 University of Jyväskylä, Finland}
\email{matti.s.vihola@jyu.fi,
  jordan.j.franks@jyu.fi}
\keywords{Adaptive,
  approximate Bayesian computation, 
  confidence interval,
  importance sampling,
  Markov chain Monte Carlo,
  tolerance choice}

\maketitle

\begin{abstract} 
Approximate Bayesian computation allows for inference of complicated
probabilistic models with intractable likelihoods using model
simulations. The Markov chain Monte Carlo implementation of
approximate Bayesian computation is often sensitive to the tolerance 
parameter: low tolerance leads to poor mixing and large tolerance
entails excess bias. We consider an approach using a relatively large
tolerance for the Markov chain Monte Carlo sampler to ensure its
sufficient mixing, and post-processing the output leading to
estimators for a range of finer tolerances. We introduce an
approximate confidence interval for the related post-corrected 
estimators, and propose an adaptive approximate Bayesian computation
Markov chain Monte Carlo, which finds a `balanced' tolerance level
automatically, based on acceptance rate optimisation. Our experiments
show that post-processing based estimators can perform better than
direct Markov chain targetting a fine tolerance, that our confidence
intervals are reliable, and that our adaptive algorithm leads to
reliable inference with little user specification.
\end{abstract} 

\section{Introduction} 
\label{sec:intro} 

Approximate Bayesian computation is a form of likelihood-free
inference \cite[see, e.g., the
reviews][]{marin-pudlo-robert-ryder,sunnaaker-busetto-numminen-corander-foll-dessimoz}
which is used when exact Bayesian inference of a parameter 
$\theta\in\mathsf{T}$ with posterior density $\pi(\theta) \propto
\mathrm{pr}(\theta)L(\theta)$ is impossible, where
$\mathrm{pr}(\theta)$ is the prior density and $L(\theta) \defeq
g(y^*\mid \theta)$ is an intractable likelihood with data
$y^*\in\mathsf{Y}$. More specifically, when the generative model of
observations $g(\uarg\mid \theta)$ cannot be evaluated, but allows for
simulations, we may perform relatively straightforward approximate inference 
based on the following (pseudo-)posterior:
\begin{equation}
    \pi_\epsilon(\theta) \propto \mathrm{pr}(\theta) L_\epsilon(\theta),
    \qquad \text{where}\qquad
    L_\epsilon(\theta) \defeq \E[ K_\epsilon(Y_\theta, y^*) ],
    \quad Y_\theta \sim g(\uarg\mid \theta),
    \label{eq:abc-post}
\end{equation}
where $\epsilon>0$ is a
`tolerance' parameter, and $K_\epsilon:\mathsf{Y}^2\to[0,\infty)$ is a `kernel' function,
which is often taken as a simple cut-off 
$K_\epsilon(y,y^*) = \charfun{\| s(y) - s(y^*)
  \| \le \epsilon}$, 
where $s:\mathsf{Y}\to\R^d$ extracts a vector of summary
statistics from the (pseudo) observations.

The summary statistics are often chosen based on the application at
hand, and reflect what is relevant for the inference task; see also
\citep{fearnhead-prangle,raynal-marin-pudlo-ribatet-robert-estoup}. 
Because $L_\epsilon(\theta)$ may be regarded as a smoothed version of
the true likelihood $g(y^*\mid \theta)$ using the kernel $K_\epsilon$,
it is intuitive that using a too large $\epsilon$ may blur the
likelihood and bias the inference. Therefore, it is generally
desirable to use as small a tolerance $\epsilon>0$ as possible, but
because the computational methods suffer from inefficiency with small
$\epsilon$, the choice of tolerance level is difficult 
\citep[cf.][]{bortot-coles-sisson,sisson-fan2018,tanaka-francis-luciani-sisson}.

We discuss simple post-processing procedure which allows for
consideration of a range of values for the tolerance $\epsilon\le
\simulationepsilon$, based on a single run of approximate Bayesian computation
Markov chain Monte Carlo \citep{marjoram-molitor-plagnol-tavare} with
tolerance $\simulationepsilon$. Such post-processing was suggested in
\citep{wegmann-leuenberger-excoffier} (in case of simple cut-off), and
similar post-processing has been suggested also with 
regression adjustment \citep{beaumont-zhang-balding} (in a rejection
sampling context).
The method, discussed further in
Section \ref{sec:method}, can be
useful for two reasons: A range of tolerances $\epsilon\le \simulationepsilon$ 
may be routinely inspected, which can reveal excess bias in 
the pseudo-posterior $\pi_{\simulationepsilon}$; and the Markov chain Monte Carlo
inference may be implemented with sufficiently large $\simulationepsilon$ to allow
for good mixing.

Our contribution is two-fold. We suggest straightforward-to-calculate
approximate confidence intervals for the posterior mean estimates
calculated from the
post-processing output, and 
discuss some theoretical properties related to it.
We also introduce an adaptive approximate
Bayesian computation Markov chain Monte Carlo
which finds a balanced $\simulationepsilon$ during
burn-in, using acceptance rate as a proxy, and detail a convergence
result for it.


\section{Post-processing over a range of tolerances}
\label{sec:method} 

For the rest of the paper, we assume that
the kernel function in \eqref{eq:abc-post} has the form
\begin{equation*}
    K_\epsilon(y, y^*)
    \defeq \phi\big( d(y,y^*)/\epsilon\big),
\end{equation*}
where $d:\mathsf{Y}^2\to[0,\infty)$ is any `dissimilarity' 
function and $\phi:[0,\infty)\to [0,1]$
is a non-increasing `cut-off' function. Typically 
$d(y,y^*) = \| s(y) - s(y^*) \|$, where $s:\mathsf{Y}^2\to\R^d$ are 
the chosen summaries, and in case of the simple cut-off discussed in
Section \ref{sec:intro}, $\phi(t) = \phi_{\mathrm{simple}}(t) \defeq \charfun{t\le
  1}$. We will implicitly assume that the pseudo-posterior $\pi_\epsilon$ given 
in \eqref{eq:abc-post} is well-defined for all $\epsilon>0$ of interest,
that is, $c_\epsilon = \int \mathrm{pr}(\theta) L_\epsilon(\theta)\ud \theta>0$.

The following summarises the approximate Bayesian computation Markov
chain Monte Carlo algorithm of \cite{marjoram-molitor-plagnol-tavare}, 
with proposal $q$ and tolerance $\delta>0$:
\begin{algo}[\abcmcmc($\simulationepsilon$)]
    \label{alg:abc-mcmc} 
Suppose $\Theta_0\in\mathsf{T}$ and $Y_0 \in\mathsf{Y}$ 
are any starting values, 
such that $\pr(\Theta_0)>0$
and $\phi( d(Y_0,y^*)/\simulationepsilon)>0$. For $k=1,2,\ldots$,
iterate:
\begin{enumerate}[(i)]
    \item Draw $\tilde{\Theta}_k \sim q(\Theta_{k-1},\uarg)$ 
      and $\tilde{Y}_k \sim g(\uarg\mid \tilde{\Theta}_{k})$.
    \item With probability $\alpha_{\delta}(\Theta_{k-1},Y_{k-1};
      \tilde{\Theta}_k,\tilde{Y}_k)$
     accept and set $(\Theta_k,Y_k) \gets
     (\tilde{\Theta}_k,\tilde{Y}_k)$; otherwise reject and set
     $(\Theta_k,Y_k) \gets
     (\Theta_{k-1},Y_{k-1})$, where
      \[
          \alpha_{\delta}(\theta,y;\tilde{\theta},\tilde{y}) \defeq 
          \min\bigg\{1,\frac{\pr(\tilde{\theta})
              q(\tilde{\theta},\theta)
              \phi\big(d(\tilde{y},y^*)/\delta\big) 
              }{
              \pr(\theta) 
                q(\theta,\tilde{\theta})
            \phi\big(d(y,y^*)/\delta\big)
            }
          \bigg\}.
     \]
\end{enumerate}
\end{algo} 
Algorithm \ref{alg:abc-mcmc} may be implemented by storing
only $\Theta_k$ and the related distances $T_k \defeq d(Y_k,y^*)$, 
and in what follows, we
regard either $(\Theta_k,Y_k)_{k\ge 1}$
or $(\Theta_k,T_k)_{k\ge 1}$ as the output of Algorithm
\ref{alg:abc-mcmc}. In practice, the initial values
$(\Theta_0,Y_0)$ should be taken as the state of the Algorithm
\ref{alg:abc-mcmc} run for a number of initial `burn-in' iterations.
We also introduce an adaptive 
algorithm for parameter tuning later (Section \ref{sec:ta}).

It is possible to consider a variant of Algorithm \ref{alg:abc-mcmc}
where many (possibly dependent) observations
$\tilde{Y}_k^{(1)},\ldots,\tilde{Y}_k^{(m)} \sim g(\uarg\mid
\tilde{\Theta}_k)$ are simulated in each iteration, and an average of
their kernel values is used in the accept-reject step
\citep[cf.][]{andrieu-lee-vihola-abc}. We focus here
in the case of single pseudo-observation per iteration, following the
asymptotic efficiency result of \cite{bornn-pillai-smith-woodard},
but remark that our method may be applied in a straightforward manner also 
with multiple observations.

\begin{definition}
    \label{def:abc-estim} 
Suppose $(\Theta_k,T_k)_{k=1,\ldots,n}$ is the output of
\abcmcmc($\simulationepsilon$) for some $\simulationepsilon>0$.
For any $\epsilon\in(0,\simulationepsilon]$ such that $\phi(T_k/\epsilon)>0$ for
some $k=1,\ldots,n$, and for any function $f:\mathsf{T}\to\R$, define
\begin{align*}
    U_k^{(\simulationepsilon,\epsilon)} &\defeq \phi(
        T_k/\epsilon)\big/\phi(T_k/\simulationepsilon), &
    \qquad W_k^{(\simulationepsilon,\epsilon)} &\defeq 
    U_k^{(\simulationepsilon,\epsilon)}\big/\textstyle\sum_{j=1}^n
      U_j^{(\simulationepsilon,\epsilon)}, \\
    E_{\simulationepsilon,\epsilon}(f) &\defeq \textstyle \sum_{k=1}^n W_k^{(\simulationepsilon,\epsilon)} f(\Theta_k), &
    S_{\simulationepsilon,\epsilon}(f) & \defeq \textstyle \sum_{k=1}^n 
    (W_k^{(\simulationepsilon,\epsilon)})^2 
    \big\{f(\Theta_k) - E_{\simulationepsilon,\epsilon}(f)\big\}^2.
\end{align*}
\end{definition} 
Algorithm \ref{alg:abc-estim} in Appendix details how
$E_{\simulationepsilon,\epsilon}(f)$ and $S_{\simulationepsilon,\epsilon}(f)$ can be
calculated in $O(n\log n)$ time simultaneously for all 
$\epsilon\le \simulationepsilon$ in case of simple cut-off.
The estimator $E_{\simulationepsilon,\epsilon}(f)$ approximates
$\E_{\pi_\epsilon}[f(\Theta)]$ and $S_{\simulationepsilon,\epsilon}(f)$ may be
used to construct a confidence interval; see Algorithm \ref{alg:ci}
below. 
Theorem \ref{thm:est-valid} details 
consistency of $E_{\simulationepsilon,\epsilon}(f)$, and relates
$S_{\simulationepsilon,\epsilon}(f)$ to the
limiting variance, in case the following (well-known) condition
ensuring a central limit theorem holds:
\begin{assumption}[Finite integrated autocorrelation]
  \label{a:iact} 
  Suppose that $\E_{\pi_{\epsilon}}[f^2(\Theta)]<\infty$ and
  $\sum_{k\ge 1} \rho_k^{(\simulationepsilon,\epsilon)}$ is finite, with
  $\rho_k^{(\simulationepsilon,\epsilon)} \defeq \mathrm{Corr}\big( h_{\simulationepsilon,\epsilon}(\Theta_0^{(s)},
  Y_0^{(s)}), h_{\simulationepsilon,\epsilon}(\Theta_k^{(s)}, Y_k^{(s)})\big)$, where
  $(\Theta_k^{(s)}, Y_k^{(s)})_{k\ge 1}$ is a stationary version
  of the \abcmcmc($\simulationepsilon$) chain, and 
  \[
      h_{\simulationepsilon,\epsilon}(\theta, y) \defeq 
      w_{\simulationepsilon,\epsilon}(y)
      f(\theta)
      \qquad\text{where}\qquad
      w_{\simulationepsilon,\epsilon}(y) \defeq 
      \phi\big(d(y,y^*)/\epsilon\big)/\phi\big(d(y,y^*)/\simulationepsilon\big).
  \]
\end{assumption}

\begin{theorem}
    \label{thm:est-valid} 
Suppose $(\Theta_k,T_k)_{k\ge 1}$ is the output of 
\abcmcmc($\simulationepsilon$),
and denote by $E_{\simulationepsilon,\epsilon}^{(n)}(f)$ and
$S_{\simulationepsilon,\epsilon}^{(n)}(f)$ 
the estimators in Definition \ref{def:abc-estim}.
If $(\Theta_k,T_k)_{k\ge 1}$ is $\varphi$-irreducible \citep{meyn-tweedie} 
then, for any $\epsilon\in(0,\simulationepsilon)$, we have 
as $n\to\infty$:
\begin{enumerate}[(i)]
\item \label{item:consistency} 
  $E_{\simulationepsilon,\epsilon}^{(n)}(f) \to
  \E_{\pi_\epsilon}[f(\Theta)]$ almost surely, whenever the
  expectation is finite.
\item \label{item:limiting-clt}
  Under Assumption \ref{a:iact}, 
  $n^{1/2}\big(E_{\simulationepsilon,\epsilon}^{(n)}(f) -
  \E_{\pi_\epsilon}[f(\Theta)]\big)\to N\big(0,v_{\simulationepsilon,\epsilon}(f)
  \tau_{\simulationepsilon,\epsilon}(f)\big)$ in
  distribution, where $\tau_{\simulationepsilon,\epsilon}(f) \defeq 
      \big(1 + 2\sum_{k\ge 1} \rho_k^{(\simulationepsilon,\epsilon)}\big)
      \in[0,\infty)$ and $n S_{\simulationepsilon,\epsilon}^{(n)}(f) \to
  v_{\simulationepsilon,\epsilon}(f)\in [0,\infty)$ almost surely.
\end{enumerate}
\end{theorem} 
Proof of Theorem \ref{thm:est-valid} is given in Appendix.
Inspired by Theorem \ref{thm:est-valid}, we suggest to report the
following approximate confidence intervals for the
suggested estimators:
\begin{algo} 
    \label{alg:ci} 
Suppose $(\Theta_k,T_k)_{k=1,\ldots,n}$ is the output of
\abcmcmc($\simulationepsilon$) and $f:\Theta\to\R$ is a function,
then for any $\epsilon\le \simulationepsilon$:
\begin{enumerate}[(i)]
    \item Calculate $E_{\simulationepsilon,\epsilon}(f)$ and
      $S_{\simulationepsilon,\epsilon}(f)$ 
      as in Definition
      \ref{def:abc-estim} (or in Algorithm \ref{alg:abc-estim}).
    \item \label{item:iact} Calculate $\hat{\tau}_{\simulationepsilon}(f)$, an
      estimate of
      the integrated autocorrelation of $\big(f(\Theta_k)\big)_{k=1,\ldots,n}$. 
    \item \label{item:ci} Report the confidence interval
      \[
          \big[ E_{\simulationepsilon,\epsilon}(f) \pm
          z_q  \big(S_{\simulationepsilon,\epsilon}(f)
          \hat{\tau}_{\simulationepsilon}(f)\big)^{1/2}
          \big],
      \]
      where $z_q>0$ corresponds to the desired normal quantile.
\end{enumerate}
\end{algo} 
The confidence interval in Algorithm \ref{alg:ci} is straightforward
application of Theorem \ref{thm:est-valid}, except for using a common
integrated autocorrelation estimate 
$\hat{\tau}_{\simulationepsilon}(f)$ for all
$\tau_{\simulationepsilon,\epsilon}(f)$. This relies on the
approximation $\tau_{\simulationepsilon,\epsilon}(f) \lessapprox
\tau_{\simulationepsilon}(f)$, which may not always be entirely
accurate, but likely to be reasonable, as illustrated by Theorem
\ref{thm:simple-acf} in Section \ref{sec:theoretical} below.  We
suggest using a common $\hat{\tau}_{\simulationepsilon}(f)$ for all
tolerances because direct estimation of integrated autocorrelation is
computationally demanding, and likely to be unstable for small
$\epsilon$. 

The classical choice for $\hat{\tau}_{\simulationepsilon}(f)$ in Algorithm
\ref{alg:ci}\eqref{item:iact} is windowed
autocorrelation, $\hat{\tau}_{\simulationepsilon}(f) = \sum_{k=-\infty}^\infty
\omega(k)\hat{\rho}_k$, with some $0\le \omega(k)\le 1$, where $\hat{\rho}_k$ is
the sample autocorrelation of $\big(f(\Theta_k)\big)$
\cite[cf.][]{geyer-practical}. We employ this approach in our
experiments with $\omega(k) =
\charfun{|k|\le M}$ where the cut-off lag
$M$ is chosen adaptively as 
the smallest integer such that $M \ge 5\big(
1 + 2 \sum_{i=1}^M \hat{\rho}_k\big)$ \citep{sokal-notes}.
Also more sophisticated techniques
for the calculation of the asymptotic variance have been suggested
\cite[e.g.][]{flegal-jones}.

We remark that, although we focus here on the case of using a common
cut-off $\phi$ for both the \abcmcmc($\simulationepsilon$) and the 
post-correction, one could also use a different cut-off $\phi_s$
in the simulation phase, as considered by
\cite{beaumont-zhang-balding} in the regression context. The extension
to Definition \ref{def:abc-estim} is straightforward, 
setting $U_k^{(\simulationepsilon,\epsilon)} \defeq \phi(
        T_k/\epsilon)\big/\phi_s(T_k/\simulationepsilon)$, 
and Theorem \ref{thm:est-valid} remains valid under a
support condition.


\section{Theoretical justification} 
\label{sec:theoretical} 

The following result, whose proof is given in Appendix, 
gives an expression for the integrated
autocorrelation in case of simple cut-off.

\begin{theorem}
    \label{thm:simple-acf} 
    Suppose Assumption \ref{a:iact} holds and
    $\phi=\phi_\mathrm{simple}$, then
    \[
        \tau_{\simulationepsilon,\epsilon}(f) - 1 
        = \frac{ \big(\check{\tau}_{\simulationepsilon,\epsilon}(f)-1\big)
        \var_{\pi_{\simulationepsilon}}(f_{\simulationepsilon,\epsilon}) +
        2 \int \pi_{\simulationepsilon}(\theta)
        \bar{w}_{\simulationepsilon,\epsilon}(\theta)
        \big(1-\bar{w}_{\simulationepsilon,\epsilon}(\theta)\big)
        \frac{r_{\simulationepsilon}(\theta)}{1-r_{\simulationepsilon}(\theta)} f^2(\theta) \ud \theta}{
        \var_{\pi_{\simulationepsilon}}(f_{\simulationepsilon,\epsilon}) 
        +\int \pi_{\simulationepsilon}(\theta)
        \bar{w}_{\simulationepsilon,\epsilon}(\theta)
        \big(1-\bar{w}_{\simulationepsilon,\epsilon}(\theta)\big)
        f^2(\theta) \ud \theta },
    \]
    where $\bar{w}_{\simulationepsilon,\epsilon}(\theta) \defeq
    L_\epsilon(\theta)/L_{\simulationepsilon}(\theta)$, 
    $f_{\simulationepsilon,\epsilon}(\theta) \defeq f(\theta)
    \bar{w}_{\simulationepsilon,\epsilon}(\theta)$,
    $\check{\tau}_{\simulationepsilon,\epsilon}(f)$ is the integrated autocorrelation 
    of $\{f_{\simulationepsilon,\epsilon}(\Theta_k^{(s)})\}_{k\ge 1}$ and
    $r_{\simulationepsilon}(\theta)$ is the rejection
    probability of the \abcmcmc($\simulationepsilon$) chain at $\theta$.
\end{theorem} 

We next discuss how this loosely suggests that $\tau_{\simulationepsilon,\epsilon}(f) \lessapprox
\tau_{\simulationepsilon,\simulationepsilon}(f) =
\tau_{\simulationepsilon}(f)$.
The weight $\bar{w}_{\simulationepsilon,\simulationepsilon}\equiv 1$, 
and under suitable
regularity conditions both $\bar{w}_{\simulationepsilon,\epsilon}(\theta)$ and
$\check{\tau}_{\simulationepsilon,\epsilon}(f)$ are continuous with respect to 
$\epsilon$,
and $\bar{w}_{\simulationepsilon,\epsilon}(\theta)\to 0$ as $\epsilon\to 0$.
Then, for $\epsilon\approx\simulationepsilon$, we have
$\bar{w}_{\simulationepsilon,\epsilon}\approx 1$ and therefore
$\tau_{\simulationepsilon,\simulationepsilon}(f)
\approx \tau_{\simulationepsilon,\epsilon}(f)$. For small $\epsilon$,
the terms with
$\var_{\pi_{\simulationepsilon}}(f_{\simulationepsilon,\epsilon})$ are of order
$O(\bar{w}_{\simulationepsilon,\epsilon}^2)$, and are dominated by the other terms 
of order $O(\bar{w}_{\simulationepsilon,\epsilon})$. The remaining ratio may be written as
\begin{align*}
    \frac{
        2 \int \pi_{\simulationepsilon}(\theta)
        \bar{w}_{\simulationepsilon,\epsilon}(\theta)
        \big(1-\bar{w}_{\simulationepsilon,\epsilon}(\theta)\big)
        \frac{r_{\simulationepsilon}(\theta)}{1-r_{\simulationepsilon}(\theta)} f^2(\theta) \ud \theta}{
        \int \pi_{\simulationepsilon}(\theta)
        \bar{w}_{\simulationepsilon,\epsilon}(\theta)
        \big(1-\bar{w}_{\simulationepsilon,\epsilon}(\theta)\big)
        f^2(\theta) \ud \theta }
   &= 2\E_{\pi_{\simulationepsilon}}\Big[ \bar{g}_{\simulationepsilon,\epsilon}^{2}(\Theta)
   \frac{r_{\simulationepsilon}(\Theta)}{1-r_{\simulationepsilon}(\Theta)}\Big] ,
\end{align*}
where $\bar{g}_{\simulationepsilon,\epsilon} \propto
\{\bar{w}_{\simulationepsilon,\epsilon}(1-\bar{w}_{\simulationepsilon,\epsilon})\}^{1/2}
f$ with $\pi_{\simulationepsilon}(\bar{g}_{\simulationepsilon,\epsilon}^2)=1$. If
$r_{\simulationepsilon}(\theta)\le
r_*<1$, then the term is upper bounded by $2 r_*(1-r_*)^{-1}$, and we
believe it to be often less than $\tau_{\simulationepsilon,\simulationepsilon}(f)$,
because the latter expression is similar to the
contribution of rejections to the integrated autocorrelation;
see the proof of Theorem \ref{thm:simple-acf}.

For general $\phi$, it appears to be hard to obtain similar
theoretical result, but we expect the approximation to be still
sensible. Theorem \ref{thm:simple-acf} relies on $Y_k^{(s)}$ being
independent of $(\Theta_k^{(0)}, Y_k^{(0)})$ conditional on
$\Theta_k^{(s)}$, assuming at least single acceptance. This is not
true with other cut-offs, but we believe that the dependence of
$Y_k^{(s)}$ from $(\Theta_0^{(s)},Y_0^{(s)})$ given $\Theta_k^{(s)}$
is generally weaker than dependence of $\Theta_k^{(s)}$ and
$\Theta_0^{(s)}$, suggesting similar behaviour.

We conclude the section with a general (albeit pessimistic) 
upper bound for the asymptotic variance of the post-corrected estimators.
\begin{theorem} 
    \label{thm:peskun-bound} 
  For any $\epsilon\le \simulationepsilon$,
  denote by $\sigma_{\simulationepsilon, \epsilon}^2(f) \defeq
  v_{\simulationepsilon,\epsilon}(f)
  \tau_{\simulationepsilon,\epsilon}(f)$
  the asymptotic variance of the estimator of Definition 
  \ref{def:abc-estim} (see Theorem  \ref{thm:est-valid}\eqref{item:limiting-clt})
  and $\bar{f}(\theta) = f(\theta) -
  \E_{\pi_{\epsilon}}[f(\Theta)]$, then for any $\epsilon \le \simulationepsilon$, 
  \[
      \sigma_{\simulationepsilon, \epsilon}^2(f) \le
      (c_{\simulationepsilon}/c_\epsilon)
      \big\{
      \sigma_{\epsilon}^2(f)
      + \tilde{\pi}_{\epsilon}\big(
      \bar{f}^2(1-w_{\simulationepsilon,\epsilon})
      \big)
      \big\},
  \]
  where $\tilde{\pi}_{\epsilon}$ is the stationary distribution of
  the direct \abcmcmc($\epsilon$) and $\sigma_{\epsilon}^2(f) \defeq
  \sigma_{\epsilon,\epsilon}^2(f)$ its asymptotic variance.
\end{theorem} 
Theorem \ref{thm:peskun-bound} follows directly from \cite[Corollary 4]{franks-vihola}.
The upper bound guarantees that a
moderate correction, that is, $\epsilon$ close to
$\simulationepsilon$ and $c_{\simulationepsilon}$ close to $c_\epsilon$, 
is nearly as efficient as direct
\abcmcmc($\simulationepsilon$). Indeed, typically $w_{\simulationepsilon,\epsilon}\to 1$ and
$c_\epsilon\to c_{\simulationepsilon}$ as $\epsilon \to
\simulationepsilon$, in which case Theorem \ref{thm:peskun-bound} implies
$\limsup_{\epsilon\to\simulationepsilon}
\sigma_{\simulationepsilon,\epsilon}^2(f) \le
\sigma_{\epsilon}^2(f)$.  However, as $\epsilon\to 0$, the bound
becomes less informative.


\section{Tolerance adaptation}
\label{sec:ta} 

We propose Algorithm \ref{alg:ta} below to adapt the tolerance
$\simulationepsilon$ in \abcmcmc($\simulationepsilon$) during a
burn-in of length $n_b$, in order to obtain a user-specified overall
acceptance rate
$\alpha^*\in (0,1)$. Tolerance optimisation has been suggested earlier
based on quantiles of distances, with parameters simulated from the
prior \citep[e.g.][]{beaumont-zhang-balding,wegmann-leuenberger-excoffier}.
This heuristic might not be satisfactory in the Markov chain Monte
Carlo context, if the prior is relatively uninformative.  We believe
that acceptance rate optimisation is a more natural alternative, and 
\cite{sisson-fan2018} suggested this as well.

Our method requires also a sequence of decreasing positive step sizes 
$(\gamma_k)_{k\ge 1}$. 
We used $\alpha^*=0.1$ and
$\gamma_k=k^{-2/3}$ in our experiments, and discuss these choices later. 

\begin{algo}
    \label{alg:ta}
Suppose $\Theta_0\in\mathsf{T}$ is a starting value with $\pr(\Theta_0)>0$.
Initialise $\simulationepsilon \defeq d(Y_0,y^*)>0$ where $Y_0 \sim
g(\uarg\mid \Theta_0)$. For $k=1,\ldots, n_b	$, iterate:
\begin{enumerate}[(i)]
\item Draw $\tilde{\Theta}_k \sim q( \Theta_{k-1},\uarg)$ and
  $\tilde{Y}_k \sim g(\uarg\mid \tilde{\Theta}_k)$.
\item With probability $A_k = \alpha_{\simulationepsilon_{k-1}}(\Theta_{k-1},Y_{k-1};
    \tilde{\Theta}_k,\tilde{Y}_k)$
  accept and set $(\Theta_k,Y_k) \gets
   (\tilde{\Theta}_k,\tilde{Y}_k)$; otherwise reject and set
    $(\Theta_k,Y_k) \gets
   (\Theta_{k-1},Y_{k-1})$.
\item $\log \simulationepsilon_{k} \leftarrow \log \simulationepsilon_{k-1} + \gamma_{k} (\alpha^* - A_k)$.
\end{enumerate}
\end{algo}

In practice, we use Algorithm \ref{alg:ta}
with a Gaussian symmetric random walk proposal $q_{\Sigma_k}$, where
the covariance parameter $\Sigma_k$ is adapted simultaneously 
\citep{haario-saksman-tamminen,andrieu-moulines} (see Algorithm
\ref{alg:ta-am} of Supplement \ref{app:ta-am}).
We only detail theory for Algorithm
\ref{alg:ta}, but note that similar simultaneous adaptation 
has been discussed earlier
\cite[cf.][]{andrieu-thoms}, and expect that 
our results could be elaborated accordingly.

The following conditions suffice for convergence of the
adaptation:
\begin{assumption}  \label{a:uniform}
  Suppose $\phi=\phi_{\mathrm{simple}}$ and the following hold:
  \begin{enumerate}[(i)]
  \item \label{item:step-size} $\gamma_k\defeq C k^{-r}$ with $r\in (\frac{1}{2}, 1]$  and $C>0$ a constant.
\item The domain $\mathsf{T}\subset\R^{n_\theta}$, $n_\theta\ge 1$, is a nonempty 
  open set and $\pr(\theta)$ is bounded.
\item \label{a:uniform-q}
  The proposal $q$ is bounded and bounded away from 
  zero.
\item The distances $D_\theta = d(Y_\theta,y^*)$ where $Y_\theta \sim
  g(\uarg\mid \theta)$ admit densities which are uniformly bounded in
  $\theta$.
\item $(\simulationepsilon_k)_{k\ge 1}$ stays in a set $[a,b]$ almost surely, where $0<a \le b <+\infty$.
\item $c_\epsilon = \int \pr(\ud \theta) L_\epsilon(\theta)>0$ for all $\epsilon\in [a,b]$.  
  \end{enumerate}
\end{assumption}

\begin{theorem}\label{thm:ta-uniform}
  Under Assumption \ref{a:uniform}, the expected value of the acceptance 
  probability, with respect to the stationary distribution of the chain, 
  converges to $\alpha^*$.
\end{theorem}

Proof of Theorem \ref{thm:ta-uniform} will follow from the more
general Theorem \ref{thm:ta-geometric} of Supplement
\ref{app:ta-geometric}.  

Polynomially decaying step size sequences as in Assumption
\ref{a:uniform} \eqref{item:step-size} are common in adaptation which
is of the stochastic approximation type as our approach
\citep{andrieu-thoms}. Slower decaying step sizes such as $n^{-2/3}$
often behave better with acceptance rate adaptation 
\citep[cf.][Remark 3]{vihola-ram}. 

Simple random walk Metropolis with covariance adaptation 
\citep{haario-saksman-tamminen} typically leads to a limiting
acceptance rate around $0.234$ \citep{roberts-gelman-gilks}. In case
of a pseudo-marginal algorithm such as \abcmcmc($\delta$), the
acceptance rate is lower than this, 
and decreases when $\delta$ is decreased (see Lemma
\ref{lem:monotonicity} of Supplement \ref{app:ta}).
Markov chain Monte Carlo would typically be necessary when rejection
sampling is not possible, that is, when the prior is far 
from the posterior. In such a case, the likelihood approximation
must be accurate enough to provide reasonable approximation 
$\pi_{\simulationepsilon} \approx \pi_{\epsilon}$.
This suggests that the desired acceptance rate should be
taken substantially lower than $0.234$. 

The choice of the desired acceptance rate $\alpha^*$ could also be
motivated by theory developed for pseudo-marginal Markov chain Monte
Carlo algorithms. \cite{doucet-deligiannidis-kohn} 
rely on log-normality of the likelihood
estimators, which is problematic in our context, because the
likelihood estimators take value zero.
\cite{sherlock-thiery-roberts-rosenthal} find the
acceptance rate $0.07$ to be optimal under certain conditions, but also in 
a quite dissimilar context. Indeed, in our context, 
the $0.07$ guideline assumes a fixed tolerance, and informs about 
choosing the number 
of pseudo-data per iteration. As we stick with single pseudo-data per
iteration following \citep{bornn-pillai-smith-woodard}, the $0.07$
guideline cannot be taken too informative. We recommend slightly 
higher $\alpha^*$ such as $0.1$ to ensure sufficient mixing.

 
\section{Post-processing with regression correction} 
\label{sec:regression} 

\cite{beaumont-zhang-balding} suggested similar post-processing as in
Section \ref{sec:method}, applying a further regression correction. 
Namely, in the context of Section \ref{sec:method}, we may consider
a function $\tilde{f}^{(\epsilon)}(\theta,y) = f(\theta) -
 \bar{s}(y)^\transpose b_{\epsilon}$
where $\bar{s}(y) = s(y) - s(y^*)$
and $b_\epsilon$ is a solution of 
\[
    \min_{a_\epsilon,b_\epsilon} 
    \E_{\tilde{\pi}_\epsilon}\big[ \big\{f(\Theta) -
      a_\epsilon - \bar{s}(Y)^\transpose b_\epsilon \big\}^2 \big]
    = \min_{a_\epsilon,b_\epsilon} \E_{\tilde{\pi}_{\simulationepsilon}}\big[
    w_{\simulationepsilon,\epsilon}(Y) \big\{f(\Theta) -
      a_\epsilon - \bar{s}(Y)^\transpose b_\epsilon \big\}^2 \big],
\]
where  $\tilde{\pi}_\simulationepsilon$
is the stationary distribution of
\abcmcmc($\simulationepsilon$), with marginal $\pi_\simulationepsilon$,
given in Appendix.
When the latter expectation is replaced by its empirical version, 
the solution coincides with weighted least squares
$(\hat{a}_{\epsilon},\hat{b}_{\epsilon}^\transpose)^\transpose
\defeq (\mat{M}^\transpose \mat{W}_\epsilon \mat{M})^{-1} \mat{M}^\transpose
\mat{W}_\epsilon v$, with
$v_k = f(\Theta_k)$, $\mat{W}_\epsilon =
\mathrm{diag}(W_1^{(\simulationepsilon,\epsilon)},\ldots,W_n^{(\simulationepsilon,\epsilon)})$
and with matrix $\mat{M}$ having rows $[M]_{k,\uarg}= (1,
\bar{s}(Y_k)^\transpose)$.

We suggest the following confidence 
interval for $a_\epsilon =
E_{\tilde{\pi}_\epsilon}[\tilde{f}^{(\epsilon)}(\Theta,Y)]$ in the spirit of Algorithm \ref{alg:ci}:
\[
    \big[ \hat{a}_\epsilon \pm
      z_q  \big(S_{\simulationepsilon,\epsilon}^{\mathrm{reg}}
      \hat{\tau}^{\mathrm{reg}}_{\simulationepsilon} \big)^{1/2}
      \big],
\]
where $\hat{\tau}^{\mathrm{reg}}_{\simulationepsilon}$ is the integrated
autocorrelation estimate for $(\hat{F}_k^{(\simulationepsilon)})$ where
$\hat{F}_k^{(\simulationepsilon)} = f(\Theta_k) - \bar{s}^T
\hat{b}_\simulationepsilon$ and
$S_{\simulationepsilon,\epsilon}^\mathrm{reg} = [(\mat{M}^\transpose \mat{W}_\epsilon
M)^{-1}]_{1,1} \sum_{k=1}^n (W_k^{(\simulationepsilon,\epsilon)})^2 (
\hat{F}_k^{(\epsilon)} - \hat{a}_\epsilon)^2$, where the first term is
included as an attempt to account for the increased uncertainty due to estimated
$\hat{b}_\epsilon$, analogous to weighted least squares. 
Experimental results show some promise for
this confidence interval, but we stress that we do not have better
theoretical backing for it, and leave further elaboration of the
confidence interval for future research.


\section{Experiments} 
\label{sec:exp} 


We experiment with our methods on two models, a lightweight Gaussian
toy example, and a Lotka-Volterra model. Our experiments focus on
three aspects: can \abcmcmc($\simulationepsilon$) with larger tolerance
$\simulationepsilon$ and post-correction to a desired tolerance
$\epsilon<\simulationepsilon$ deliver more accurate results than direct
\abcmcmc($\epsilon$); does the approximate confidence interval appear
reliable; how well does the tolerance adaptation work in practice.
All the experiments are implemented in Julia \citep{julia}, and the
codes are available in \url{https://bitbucket.org/mvihola/abc-mcmc}.

Because we believe that Markov chain Monte Carlo is most useful when
little is known about the posterior, we apply covariance adaptation 
\citep{haario-saksman-tamminen,andrieu-moulines} throughout 
the simulation in all our experiments, using an identity covariance initially.
When running the covariance adaptation alone, we employ the step size 
$n^{-1}$ as in the original method of \cite{haario-saksman-tamminen},
and in case of tolerance adaptation,
we use step size $n^{-2/3}$.

Regarding our first question, we investigate running
\abcmcmc($\simulationepsilon$) starting near the posterior mode with different
pre-selected tolerances $\simulationepsilon$.
We first attempted to perform the experiments by
initialising the chains from independent samples of the prior
distribution, but in this case, most of the chains failed to accept a
single move during the entire run. In contrast, our experiments with
tolerance adaptation are initialised from the prior,
and both the tolerances and the covariances are adjusted
fully automatically by our algorithm. 


\subsection{One-dimensional Gaussian model} 
\label{sec:gauss-exp} 

Our first model is a toy model with $\pr(\theta) = N(\theta; 0, 30^2)$,
$g(y\mid \theta) = N(y; \theta, 1)$ 
and $d(y,y^*) = |y|$.
The true posterior without approximation is Gaussian.  While this
scenario is clearly academic, the prior is far from the posterior,
making rejection sampling approximate Bayesian computation
inefficient. It is clear that $\pi_\epsilon$ has zero mean for all
$\epsilon$ (by symmetry), and that $\pi_\epsilon$ is more spread for
bigger $\epsilon$. We experiment with both simple cut-off
$\phi_\mathrm{simple}$ and Gaussian cut-off $\phi_{\mathrm{Gauss}}(t)
\defeq e^{-t^2/2}$.

We run the experiments with 10,000 independent chains, each for 11,000
 iterations including 1,000 burn-in. The chains were always started
 from $\theta_0 = 0$. We inspect estimates for the posterior mean 
 $\E_{\pi_\epsilon}[f(\Theta)]$ for $f(\theta) = \theta$ and
 $f(\theta) = |\theta|$. Figure \ref{fig:gauss-simple} (left) shows
 the estimates with their confidence intervals based on a
 single realisation of \abcmcmc(3). Figure \ref{fig:gauss-simple}
 (right) shows box plots of the estimates calculated from each
 \abcmcmc($\simulationepsilon$), with $\simulationepsilon$ indicated
 by colour; the rightmost box plot (blue) corresponds to \abcmcmc(3),
 the second from the right (red) \abcmcmc(2.275) etc. For
 $\epsilon=0.1$, the post-corrected estimates from \abcmcmc($0.825$)
 and \abcmcmc($1.55$) appear slightly more accurate than direct
 \abcmcmc($0.1$).  Similar figure for Gaussian cut-off, with similar
 findings, may be found in the Supplement Figure
 \ref{fig:gauss-gauss}.
\begin{figure} 
\includegraphics[height=5cm]{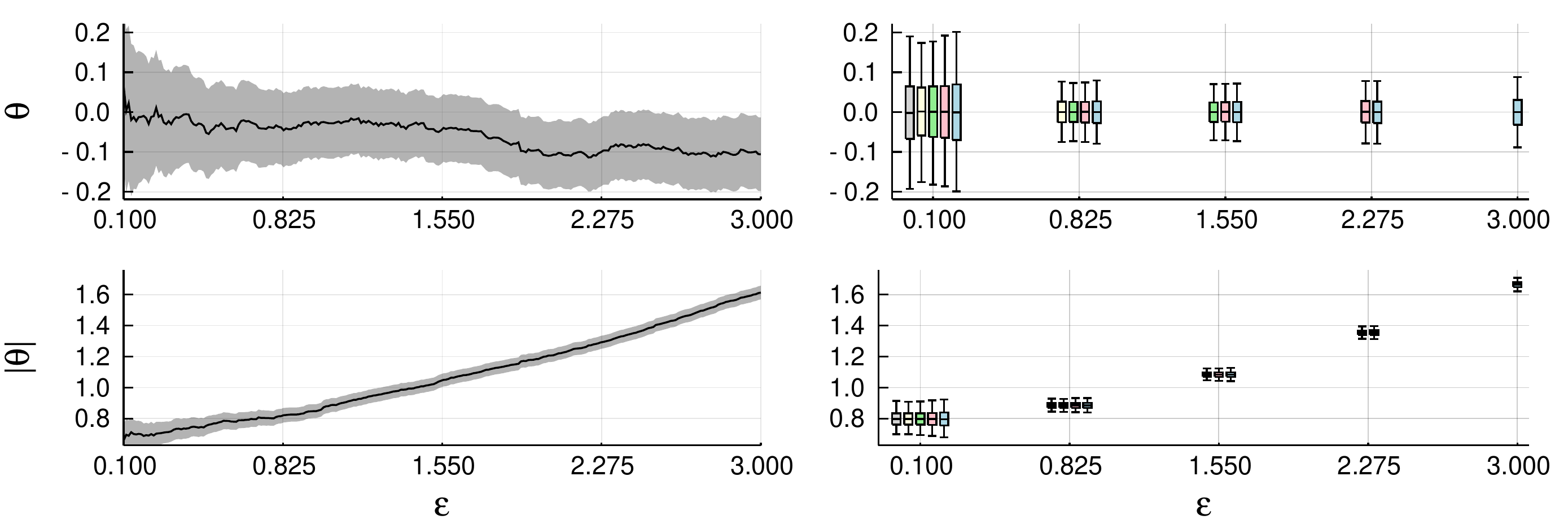} 
\caption{Gaussian model with $\phi_\mathrm{simple}$.
  Estimates from single run of \abcmcmc($3$) (left) and
  estimates from 10,000 replications of
  \abcmcmc($\simulationepsilon$) for $\simulationepsilon\in\{0.1, 0.825, 1.55, 2.275,
    3\}$ indicated by colours.
  } 
\label{fig:gauss-simple}
\end{figure}

Table \ref{tab:gauss-confidence} shows frequencies of the
calculated 95\% confidence intervals containing the `ground truth',
as well as mean acceptance rates.
The ground truth for $\E_{\pi_\epsilon}[f_1(\Theta)]$ is known to be
zero for all $\epsilon$, and the overall mean of all the calculated
estimates is used as the ground truth for
$\E_{\pi_\epsilon}[f_2(\Theta)]$.  The frequencies appear close to
ideal with the post-correction approach, being slightly pessimistic in
case of simple cut-off as anticipated by the theoretical
considerations (cf.~Theorem \ref{thm:simple-acf} and the 
discussion below).

\begin{table} 
    \caption{Frequencies of the 95\% confidence intervals, from
      \abcmcmc($\delta$) to tolerances $\epsilon$, containing
      the ground truth in the Gaussian model.}
    \label{tab:gauss-confidence} 
     \small
    \begin{center}
    \begin{tabular}{llccccccccccc}
        \toprule
        \small
        &&\multicolumn{5}{c}{$f(x)=x$}
        &\multicolumn{5}{c}{$f(x)=|x|$}
        & \raisebox{-1ex}{Acc.}\\
        \cmidrule(lr){3-7}
        \cmidrule(lr){8-12}
        Cut-off&\raisebox{-1.5pt}{$\simulationepsilon$} 
        $\backslash$
        \raisebox{3pt}{$\epsilon$}
        &0.10 & 0.82 & 1.55 & 2.28 & 3.00  
        &0.10 & 0.82 & 1.55 & 2.28 & 3.00  & rate\\
        \midrule
        \multirow{5}{*}{$\phi_{\mathrm{simple}}$}
& 0.1 & 0.93 &      &      &      &      & 0.93 &      &      &      &
   & 0.03\\
   & 0.82 & 0.97 & 0.95 &      &      &      & 0.95 & 0.94 &      &
&      & 0.22\\
& 1.55 & 0.97 & 0.97 & 0.95 &      &      & 0.96 & 0.95 & 0.95 &
&      & 0.33\\
& 2.28 & 0.98 & 0.97 & 0.96 & 0.95 &      & 0.96 & 0.96 & 0.96 & 0.95
&      & 0.4\\
& 3.0 & 0.98 & 0.98 & 0.97 & 0.97 & 0.95 & 0.96 & 0.96 & 0.96 & 0.95 &
0.95 & 0.43\\
 \midrule
        \multirow{5}{*}{$\phi_{\mathrm{Gauss}}$} 
        & 0.1 & 0.93 &      &      &      &      & 0.93 &      &
        &      &      & 0.05\\
        & 0.82 & 0.94 & 0.95 &      &      &      & 0.92 & 0.95 &
        &      &      & 0.29\\
        & 1.55 & 0.94 & 0.94 & 0.95 &      &      & 0.94 & 0.94 & 0.95
        &      &      & 0.38\\
        & 2.28 & 0.95 & 0.95 & 0.95 & 0.95 &      & 0.95 & 0.95 & 0.96
        & 0.95 &      & 0.41\\
        & 3.0 & 0.95 & 0.95 & 0.95 & 0.95 & 0.95 & 0.95 & 0.96 & 0.95
        & 0.95 & 0.95 & 0.42\\
        \bottomrule
\end{tabular}
\end{center}
\end{table}

Figure \ref{fig:NormalAdaptComparison} shows progress of tolerance
adaptations during the burn-in, and histogram of the mean
acceptance rates of the chain after burn-in.
The lines on the left show the median, and the shaded regions indicate the 
50\%, 75\%, 95\% and 99\% quantiles.
The figure suggests concentration, but reveals that the adaptation has not 
fully converged yet. This
is also visible in the mean acceptance rate over all realisations,
which is $0.17$ for simple cut-off and 
$0.12$ for Gaussian cut-off (see Figure
\ref{fig:NormalAdaptComparison-GaussCutoff} in the Supplement).
Table \ref{tab:rmse-normal-adaptation} shows root mean
square errors for target tolerance
$\epsilon=0.1$, with both \abcmcmc($\simulationepsilon$) with $\simulationepsilon$ fixed
as above, and for the tolerance adaptive algorithm.
Here, only the adaptive chains with final tolerance $\ge 0.1$ were 
included (9,998 and
9,993 out of 10,000 chains for $\phi_{\mathrm{simple}}$ and
$\phi_{\mathrm{Gauss}}$, respectively).
Tolerance adaptation (started from prior distribution) appears to 
be competitive
with `optimally' tuned fixed tolerance \abcmcmc($\simulationepsilon$).

\begin{figure} 
\includegraphics[height=3.1875cm]{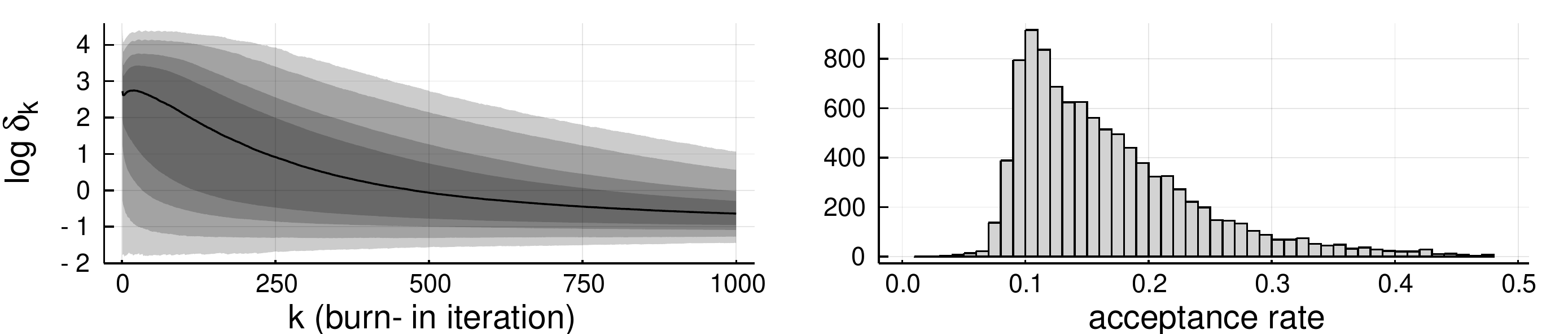} 
\caption{Progress of tolerance adaptation (left) and histogram of acceptance
  rates (right) in the Gaussian model experiment with simple cut-off.} 
\label{fig:NormalAdaptComparison}
\end{figure}

\begin{table}
\caption{Root mean square errors $(\times 10^{-2})$ from
  \abcmcmc($\simulationepsilon$) for tolerance $\epsilon=0.1$ with fixed tolerance 
  and with the adaptive algorithms in the Gaussian model.} 
    \label{tab:rmse-normal-adaptation} 
\small
\begin{center}
\begin{tabular}{lcccccccccccc} 
\toprule
& \multicolumn{6}{c}{$\phi_{\mathrm{simple}}$} 
& \multicolumn{6}{c}{$\phi_{\mathrm{Gauss}}$}  \\
    \cmidrule(lr){2-7} \cmidrule(lr){8-13} 
    & \multicolumn{5}{c}{Fixed tolerance} & Adapt  
    & \multicolumn{5}{c}{Fixed tolerance} & Adapt  \\
    \cmidrule(lr){2-6} \cmidrule(lr){7-7}
    \cmidrule(lr){8-12} \cmidrule(lr){13-13}
$\simulationepsilon$ & 0.1 & 0.82 & 1.55 & 2.28 & 3.0 & 0.64 & 0.1 & 0.82 & 1.55 & 2.28 &
3.0 & 0.28\\
\midrule
$x$ & 9.75 & 8.95 & 9.29 & 9.65 & 10.3 & 9.15 & 7.97 & 7.12 & 7.82 & 8.94
& 9.93 & 7.08\\ 
$|x|$ & 5.49 & 5.35 & 5.51 & 5.81 & 6.24 & 5.38 & 4.47 & 4.22 & 4.68 & 5.26
& 5.95 & 4.15\\ 
\bottomrule
\end{tabular}
\end{center}
\end{table}


\subsection{Lotka-Volterra model} 
\label{sec:lotkavolterra-exp} 

Our second experiment is a Lotka-Volterra model suggested by
\cite{boys-wilkinson-kirkwood},
which was considered in the approximate Bayesian computation context
by \cite{fearnhead-prangle}.
The model is a Markov process $(X_t,Y_t)_{t\ge 0}$ of counts,
corresponding to a reaction network $X\to 2X$ with rate $\theta_1$, 
$X + Y \to 2Y$
with rate $\theta_2$ and $Y \to \emptyset$ with rate $\theta_3$.
The reaction log-rates $(\log\theta_1,\log\theta_2,\log\theta_3)^\transpose$ are
parameters, which we equip with a
uniform prior, $(\log\theta_1, \log\theta_2,\log\theta_3)^\transpose\sim
U([-6,0]^3)$. The data is a simulated trajectory from the model with 
$\theta = (0.5,0.0025,0.3)^\transpose$ until time $40$. The inference is 
based on the Euclidean distances of five-dimensional summary statistics 
of the process observed every 5 time
units ($\tilde{X}_k = X_{5k}$ and $\tilde{Y}_k = Y_{5k}$). The summary
statistics are the sample autocorrelation of $(\tilde{X}_k)$ at lag 2 
multiplied by 100, and the 10\% and 90\% quantiles of $(\tilde{X}_k)$ 
and $(\tilde{Y}_k)$.
The observed summary statistics are $(-51.07, 29, 304, 65, 404)^\transpose$.

We first run comparisons similar to Section \ref{sec:gauss-exp}, but
now with 1,000 independent \abcmcmc($\simulationepsilon$) chains with 
simple cut-off. We investigate
the effect of post-correction, with 20,000 samples, including 10,000
burn-in, for each chain. All chains were started from near the
posterior mode, from $(-0.55, -5.77, -1.09)^\transpose$. Figure
\ref{fig:lotkavolterra-comparison} 
\begin{figure} 
\includegraphics[height=7cm]{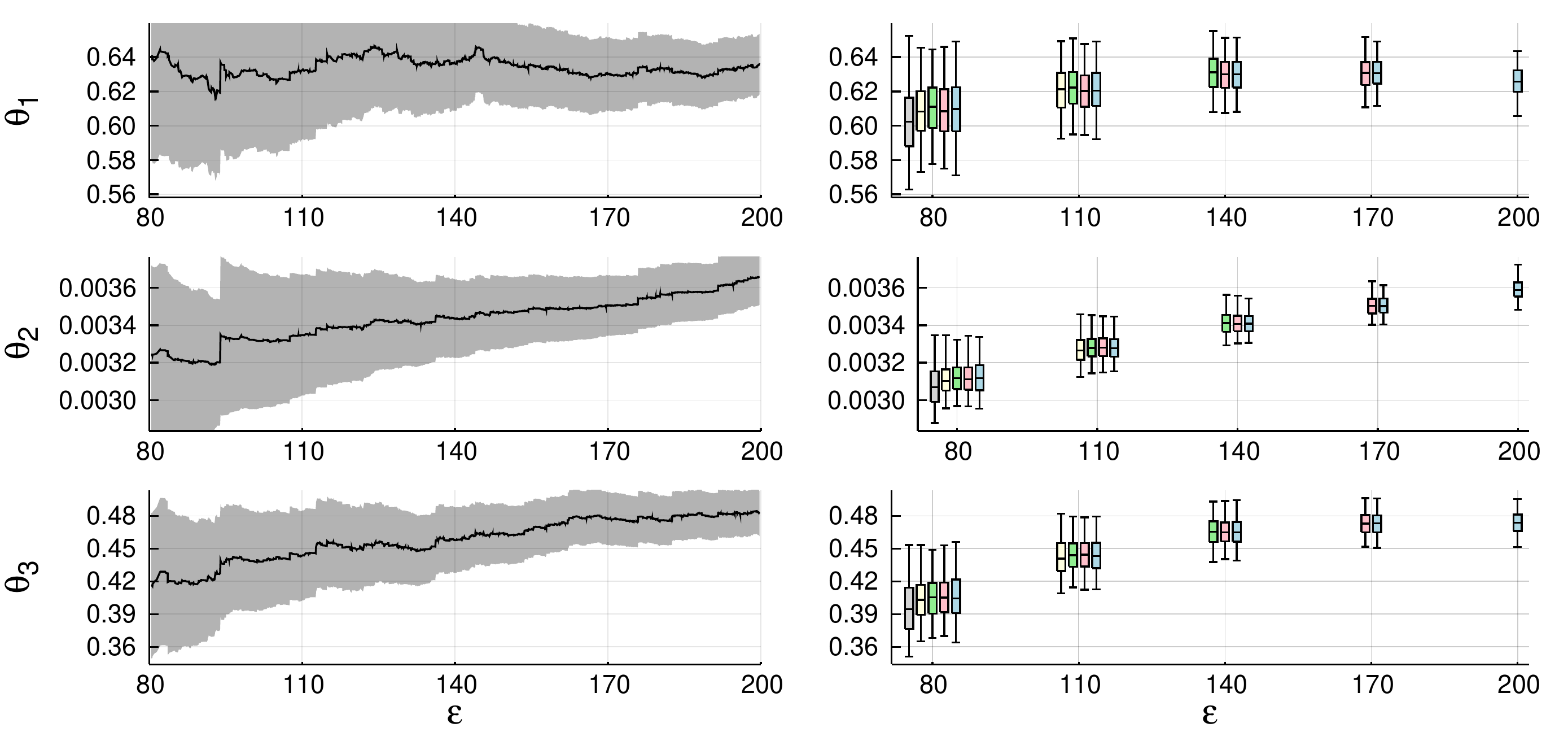}
\caption{Lotka-Volterra model with simple cut-off.
  Estimates from single run of \abcmcmc($200$) (left) and
  estimates from 1,000 replications of
  \abcmcmc($\simulationepsilon$) with $\simulationepsilon\in\{80,110,140,170,200\}$
  indicated by colour.} 
\label{fig:lotkavolterra-comparison}
\end{figure} 
shows similar comparisons as in
Section \ref{sec:gauss-exp}, and Figure
\ref{fig:lotkavolterra-epa-reg-comparison} 
\begin{figure} 
\includegraphics[height=7cm]{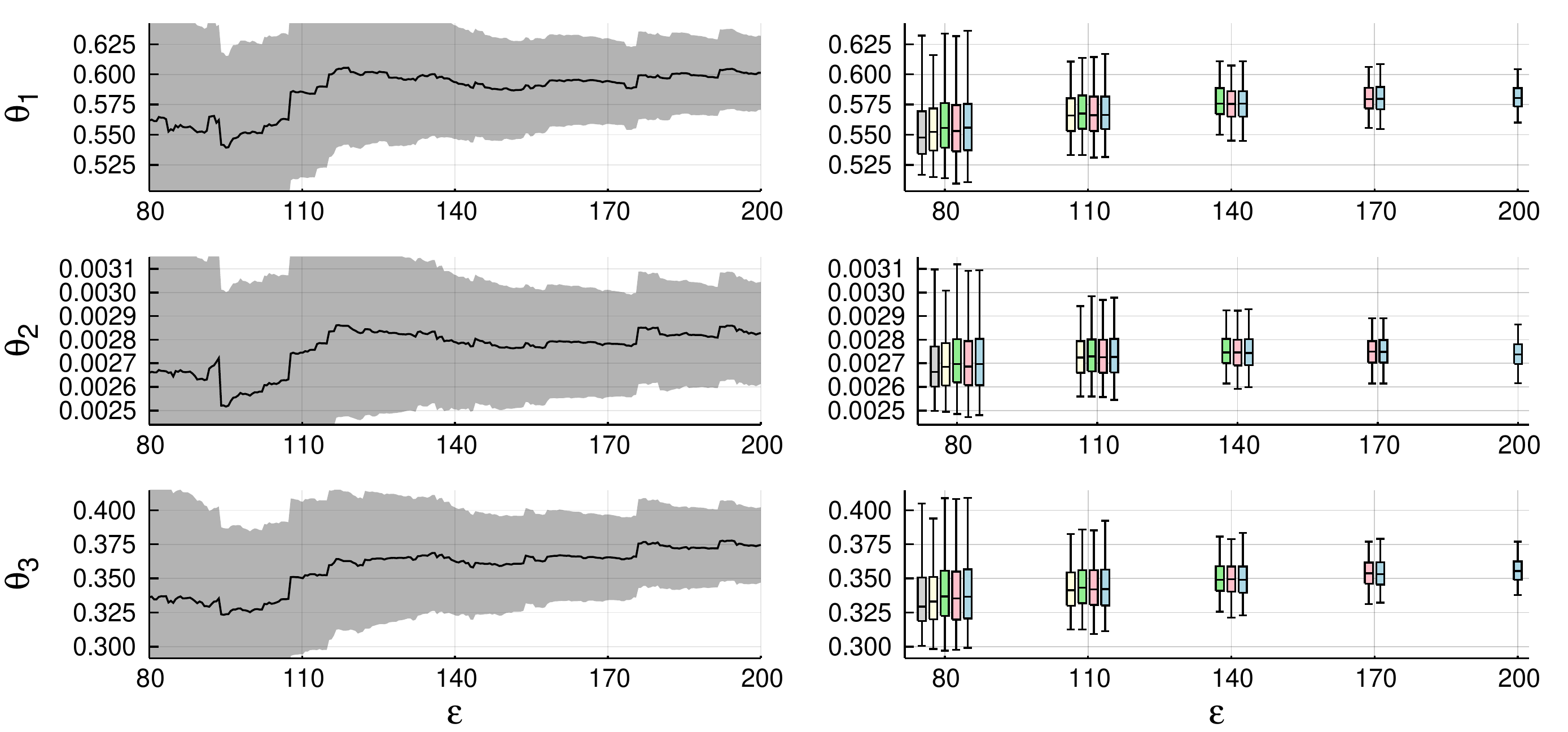}
\caption{Lotka-Volterra model with Epanechnikov cut-off and regression
  correction.
  Estimates from single run of \abcmcmc($200$) (left) and
  estimates from 1,000 replications of
  \abcmcmc($\simulationepsilon$) with $\simulationepsilon\in\{80,110,140,170,200\}$
  indicated by colour.}
\label{fig:lotkavolterra-epa-reg-comparison}
\end{figure} 
shows results for regression
correction with Epanechnikov cut-off $\phi_{\mathrm{Epa}}(t) =
\max\{0,1-t^2\}$
\citep{beaumont-zhang-balding}. 
The results suggest that
post-correction might provide slightly more accurate estimators,
particularly with smaller tolerances. There is also some bias in
\abcmcmc($\simulationepsilon$) with smaller $\simulationepsilon$, when compared to the
ground truth calculated from \abcmcmc($\simulationepsilon$) chain of
ten million iterations.
Table \ref{tab:lotkavolterra-confidence} shows coverages
of confidence intervals.
\begin{table} 
    \setlength{\tabcolsep}{0.7ex}
    \caption{Mean acceptance rates and frequencies 
      of the 95\% confidence intervals, from
      \abcmcmc($\simulationepsilon$) to tolerances $\epsilon$, in the
      Lotka-Volterra model.}
    \label{tab:lotkavolterra-confidence} 
    \small
    \begin{center}
    \begin{tabular}{llcccccccccccccccc}
        \toprule
        \small
        &
        &\multicolumn{5}{c}{$f(\theta)=\theta_1$}
        &\multicolumn{5}{c}{$f(\theta)=\theta_2$}
        &\multicolumn{5}{c}{$f(\theta)=\theta_3$}
        & \raisebox{-1ex}{Acc.}
        \\
        \cmidrule(lr){3-7}
        \cmidrule(lr){8-12}
        \cmidrule(lr){13-17}
& 
\raisebox{-1.5pt}{$\simulationepsilon$} 
        $\backslash$ \raisebox{3pt}{$\epsilon$}
& 80 & 110 & 140 & 170 & 200
& 80 & 110 & 140 & 170 & 200
& 80 & 110 & 140 & 170 & 200
& rate \\
\midrule
\multirow{5}{*}{\rotatebox{90}{$\phi_{\mathrm{simple}}$}}
& 80 & 0.8 &      &      &      &      & 0.73 &      &      &      &
   & 0.74 &      &      &      &      & 0.05\\ 
   & 110 & 0.97 & 0.93 &      &      &      & 0.94 & 0.89 &      &
 &      & 0.94 & 0.9 &      &      &      & 0.07\\ 
 & 140 & 0.99 & 0.97 & 0.93 &      &      & 0.98 & 0.96 & 0.92 &
 &      & 0.98 & 0.96 & 0.94 &      &      & 0.1\\ 
 & 170 & 0.99 & 0.98 & 0.96 & 0.93 &      & 0.98 & 0.97 & 0.96 &
 0.93 &      & 0.99 & 0.98 & 0.96 & 0.95 &      & 0.14\\ 
 & 200 & 1.0 & 0.99 & 0.98 & 0.97 & 0.94 & 0.99 & 0.99 & 0.98 & 0.97
 & 0.92 & 0.99 & 0.98 & 0.98 & 0.96 & 0.94 & 0.17\\
\midrule 
\multirow{5}{*}{\rotatebox{90}{regr. $\phi_{\mathrm{Epa}}$}}
 & 80 & 0.75 &      &      &      &      & 0.76 &      &      &
 &      & 0.68 &      &      &      &      & 0.05\\ 
 & 110 & 0.92 & 0.92 &      &      &      & 0.93 & 0.94 &      &
 &      & 0.87 & 0.91 &      &      &      & 0.07\\ 
 & 140 & 0.93 & 0.94 & 0.94 &      &      & 0.94 & 0.96 & 0.97 &
 &      & 0.9 & 0.92 & 0.94 &      &      & 0.1\\ 
 & 170 & 0.93 & 0.95 & 0.95 & 0.95 &      & 0.96 & 0.97 & 0.97 &
 0.98 &      & 0.92 & 0.94 & 0.94 & 0.95 &      & 0.14\\ 
 & 200 & 0.96 & 0.96 & 0.96 & 0.96 & 0.96 & 0.98 & 0.98 & 0.98 &
 0.98 & 0.98 & 0.95 & 0.96 & 0.95 & 0.96 & 0.96 & 0.17\\
\bottomrule
\end{tabular}
\end{center}
\end{table}

In addition, we experiment with the tolerance adaptation, using also
20,000 samples out of which 10,000 are burn-in. Figure
\ref{fig:lotkavolterraAdaptComparison} 
\begin{figure} 
\includegraphics[height=3.1875cm]{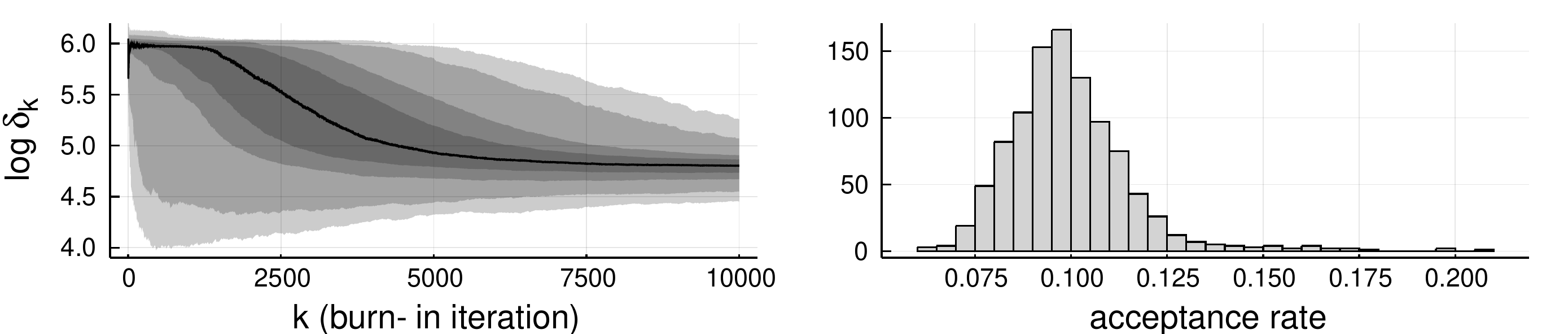}
\caption{Progress of tolerance adaptation (left) and histogram of
  acceptance rates (right) 
  in the Lotka-Volterra experiment.} 
\label{fig:lotkavolterraAdaptComparison}
\end{figure}
shows the progress of the
$\log$-tolerance during the burn-in, and histogram of the realised
mean acceptance rates during the estimation phase.  The realised
acceptance rates are concentrated around the mean $0.10$.
Table \ref{tab:rmse-lotkavolterra-adaptation} 
\begin{table}
    \caption{Root mean square errors of estimators
      from \abcmcmc($\simulationepsilon$)
      for tolerance
      $\epsilon=80$, with fixed tolerance and with
      adaptive tolerance in the Lotka-Volterra model.} 
    \label{tab:rmse-lotkavolterra-adaptation} 
\small
\begin{center}
\begin{tabular}{lcccccccccccc} 
\toprule
    & \multicolumn{6}{c}{Post-correction, simple cut-off} 
    & \multicolumn{6}{c}{Regression, Epanechnikov cut-off} \\
    \cmidrule(lr){2-7} 
    \cmidrule(lr){8-13} 
    & \multicolumn{5}{c}{Fixed tolerance} & Adapt  
    & \multicolumn{5}{c}{Fixed tolerance} & Adapt  \\
    \cmidrule(lr){2-6} \cmidrule(lr){7-7}
    \cmidrule(lr){8-12} \cmidrule(lr){13-13}
$\simulationepsilon$& 80 & 110 & 140 & 170 & 200 & 122.6 & 80 & 110 & 140
& 170 & 200 & 122.6\\ 
\midrule
$\theta_1$ $(\times 10^{-2})$& 2.37 & 1.81 & 1.75 & 1.83 & 1.93 & 1.8 & 3.1 & 2.74 & 3.02 & 3.09 &
3.19 & 2.57\\ 
$\theta_2$ $(\times 10^{-4})$ & 1.32 & 0.99 & 0.93 & 0.96 & 1.06 & 1.04 & 1.52 & 1.39 & 1.54 & 1.61
& 1.63 & 1.28\\ 
$\theta_3$ $(\times 10^{-2})$& 2.94 & 2.26 & 2.11 & 2.14 & 2.37 & 2.34 & 2.77 & 2.53 & 2.76 & 2.85
& 2.91 & 2.34\\
\bottomrule
\end{tabular}
\end{center}
\end{table} 
shows root mean square
errors of the estimators from \abcmcmc($\simulationepsilon$) for $\epsilon=80$
for fixed tolerance and with tolerance
adaptation. Only the adaptive chains with final tolerance $\ge
80.0$ were included (999 out of 1,000 chains).

In this case, the chains run with the tolerance adaptation led to
better results than those run only with the covariance 
adaptation (and
fixed tolerance). This perhaps surprising result may be due to the
initial behaviour of the covariance adaptation, which may be unstable
when there are many rejections.  Different initialisation strategies,
for instance following \citep[Remark 2]{haario-saksman-tamminen},
might lead to more stable behaviour compared to using the adaptation
of \cite{andrieu-moulines} from the start, as we do. The different
step size sequences ($n^{-1}$ and $n^{-2/3}$) could also play a
r\^{o}le. We repeated the experiment for the chains with fixed
tolerances, but now with covariance adaptation step size $n^{-2/3}$.
This led to more accurate estimators for \abcmcmc($\simulationepsilon$) 
with higher $\simulationepsilon$, but worse behaviour with smaller $\simulationepsilon$.
In any case, also here, tolerance adaptation delivered 
competitive results (see Supplement \ref{app:results}).



\section{Discussion} 
\label{sec:discussion} 

We believe that approximate Bayesian computation inference with 
Markov chain Monte Carlo is a useful approach, when the chosen simulation
tolerance allows for good mixing. Our confidence
intervals for post-processing and automatic tuning of simulation
tolerance  may make this approach more appealing in practice. 

A related approach by \cite{bortot-coles-sisson} makes tolerance an
auxiliary variable with a user-specified prior.  This approach avoids
explicit tolerance selection, but the inference is based on 
a pseudo-posterior $\check{\pi}(\theta,\delta)$ 
not directly related to
$\pi_\delta(\theta)$ in \eqref{eq:abc-post}.
\cite{bortot-coles-sisson} also provide tolerance-dependent analysis,
showing parameter means and variances with respect to conditional
distributions of $\check{\pi}(\theta,\delta)$ given 
$\delta\le\epsilon$. We
believe that our approach, where the effect of tolerance in the
expectations with respect $\pi_\epsilon$ can be investigated 
explicitly, can be more immediate to interpret. Our confidence
interval only shows the Monte Carlo uncertainty related to the
posterior mean, and we are currently investigating how the overall
parameter uncertainty could be summarised in a useful manner.

The convergence rates of approximate Bayesian computation has been
investigated by \cite{barber-voss-webster} in terms of cost and bias
with respect to true posterior, and recently by
\cite{li-fearnhead-efficiency,li-fearnhead-regression} in the large
data limit, the latter in the context of regression. It would be
interesting to consider extensions of these results in the Markov
chain Monte Carlo context. In fact, \cite{li-fearnhead-efficiency}
already suggest that the acceptance rate must be lower bounded, which
is in line with our adaptation rule.

Automatic selection of tolerance has been considered earlier in
\cite{ratmannJHS},  who propose an algorithm based on tempering and a
cooling schedule. Based on our experiments, the tolerance adaptation
we present in this paper appears to perform well in practice, and
provides reliable results with post-correction.  
For the adaptation to work efficiently,
the Markov chains must be taken relatively long, rendering the approach
difficult for the most computationally demanding models. 

We conclude with a brief discussion of certain extensions of the suggested
post-correction method; more details are given in Supplement
\ref{app:extensions}. First, in case of non-simple cut-off, the rejected
samples may be `recycled' by using the acceptance probability as weight 
\citep{ceperley-chester-kalos}.
The accuracy
of the post-corrected estimator could be enhanced with smaller values of 
$\epsilon$ by performing
further independent simulations from $g(\uarg\mid \Theta_k)$ 
(which may be calculated in
parallel). The estimator is rather straightforward, but requires some
care because the estimators of the pseudo-likelihood take value zero.
The latter extension, which involves additional simulations as
post-processing, is similar to the `lazy' version of  
\cite{prangle-lazy,prangle-lazier}
incorporating a randomised stopping rule for simulation, 
and to debiased `exact' approach of
\cite{tran-kohn-exact-abc}, which may lead to estimators which get
rid of $\epsilon$-bias entirely.


\section{Acknowledgements} 

This work was supported by Academy of Finland (grants 274740, 284513
and 312605). The authors wish to acknowledge CSC, IT Center for
Science, Finland, for computational resources, and 
thank Christophe Andrieu for useful discussions.




\appendix

\section*{Appendix}

The following algorithm shows that in case of simple (post-correction) cut-off,
$E_{\simulationepsilon,\epsilon}(f)$ and $S_{\simulationepsilon,\epsilon}(f)$ may be
calculated simultaneously for all tolerances efficiently:
\begin{algo}
    \label{alg:abc-estim} 
          Suppose $\phi = \phi_{\mathrm{simple}}$ and 
$(\Theta_k,T_k)_{k=1,\ldots,n}$ is the output of \abcmcmc($\simulationepsilon$).
\begin{enumerate}[(i)]
    \item \label{item:sort} Sort $(\Theta_k,T_k)_{k=1,\ldots,n}$ with respect to $T_k$:
      \begin{itemize}
          \item Find
      indices $I_1,\ldots,I_n$ such that $T_{I_k}\le T_{I_{k+1}}$ for all
        $k=1,\ldots,n-1$.
    \item Denote $(\hat{\Theta}_k, \hat{T}_k) \gets 
      (\Theta_{I_k}, T_{I_k})$.
      \end{itemize}
    \item For all unique values $\epsilon\in\{\hat{T}_1,\ldots,\hat{T}_n\}$, let 
      $m_\epsilon \defeq \max\{ k\ge 1\given \hat{T}_k\le \epsilon\}$, and
      define 
      \[
          E_{\simulationepsilon,\epsilon}(f) \defeq 
          m_\epsilon^{-1}
          \textstyle \sum_{k=1}^{m_\epsilon} 
          f(\hat{\Theta}_k), 
          \qquad
          \text{and}
          \qquad
          S_{\simulationepsilon,\epsilon}(f) \defeq 
          m_\epsilon^{-2} \textstyle \sum_{k=1}^{m_\epsilon}
          \big[ f(\hat{\Theta}_k) - E_{\simulationepsilon,\epsilon}(f)\big]^2.
      \]
      (and for $\hat{T}_k<\epsilon<\hat{T}_{k+1}$, let
      $E_{\simulationepsilon,\epsilon}(f) \defeq E_{\simulationepsilon,\hat{T}_k}(f)$ and 
      $S_{\simulationepsilon,\epsilon}(f) \defeq S_{\simulationepsilon,\hat{T}_k}(f)$.)
\end{enumerate}
\end{algo}
The sorting in Algorithm \ref{alg:abc-estim}\eqref{item:sort} may be
performed in $O(n\log n)$ time, and $E_{\simulationepsilon,\epsilon}(f)$ and
$S_{\simulationepsilon,\epsilon}(f)$ may all be calculated in $O(n)$ time by
forming appropriate cumulative sums.

\begin{proof}[of Theorem \ref{thm:est-valid}] 
Algorithm \ref{alg:abc-mcmc} is a Metropolis--Hastings algorithm with
compound proposal $\tilde{q}(\theta, y;
\theta',y') = q(\theta, \theta')g(y'\mid \theta')$ and with
target 
$\tilde{\pi}_{\epsilon}(\theta,y) \propto
\mathrm{pr}(\theta) g(y\mid \theta) \phi\big(
d(y,y^*)/\epsilon\big)$.
The chain $(\Theta_k,Y_k)_{k\ge 1}$ is Harris-recurrent, as 
a full-dimensional
Metropolis--Hastings which is $\varphi$-irreducible
\citep{roberts-rosenthal-harris}.
Because $\phi$ is monotone and $\epsilon\le \simulationepsilon$, we have 
$\phi\big(d(y,y^*)/\simulationepsilon\big)
    \ge \phi\big(d(y,y^*)/\epsilon\big)$,
and therefore $\tilde{\pi}_\epsilon$ is absolutely continuous with
respect to $\tilde{\pi}_{\simulationepsilon}$, and
$w_{\simulationepsilon,\epsilon}(y)
    = c_{\simulationepsilon,\epsilon}
    \tilde{\pi}_{\epsilon}(\theta,y)/
      \tilde{\pi}_{\simulationepsilon}(\theta,y)$,
where $c_{\simulationepsilon,\epsilon}>0$ is a constant.
If we denote $\xi_k(f) \defeq U_k^{(\simulationepsilon,\epsilon)} f(\Theta_k)$ and 
$\xi_k(\mathbf{1}) \defeq U_k^{(\simulationepsilon,\epsilon)} =
w_{\simulationepsilon,\epsilon}(Y_k)$, 
then
$E_{\simulationepsilon,\epsilon}^{(n)}(f) = 
    \sum_{k=1}^n \xi_k(f) /\sum_{j=1}^n \xi_j(\mathbf{1})
    \to
    \E_{\tilde{\pi}_\epsilon}[f(\Theta)]$ almost surely
by Harris recurrence and $\tilde{\pi}_\epsilon$ invariance
\citep[e.g.][]{vihola-helske-franks}.
The claim  \eqref{item:consistency} follows because $\pi_\epsilon$ is the 
marginal density of 
$\tilde{\pi}_\epsilon$.

The chain $(\Theta_k,Y_k)_{k\ge 1}$ is reversible, so
\eqref{item:limiting-clt} follows by 
\citep[Theorem 7(i)]{vihola-helske-franks},
because $m_{f}^{(2)}(\theta,y) \defeq w_{\simulationepsilon,\epsilon}^2(y) f^2(\theta)$
satisfies
\[
    \E_{\tilde{\pi}_{\simulationepsilon}}[m_{f}^{(2)}(\Theta,Y)] =
    c_{\simulationepsilon,\epsilon}
    \E_{\tilde{\pi}_{\epsilon}}[w_{\simulationepsilon,\epsilon}(Y)f^2(\Theta)]
    \le c_{\simulationepsilon,\epsilon} \E_{\pi_\epsilon}[f^2(\Theta)]<\infty,
\]
and because the asymptotic variance of the function
$h_{\simulationepsilon,\epsilon}$
with respect to $(\Theta_k,Y_k)_{k\ge 1}$ may be expressed 
as $\var_{\tilde{\pi}_{\simulationepsilon}}\big(h_{\simulationepsilon,\epsilon}(\Theta,
Y)\big)\tau_{\simulationepsilon,\epsilon}(f)$, so 
$v_{\simulationepsilon,\epsilon}(f) =
\var_{\tilde{\pi}_{\simulationepsilon}}\big(h_{\simulationepsilon,\epsilon}(\Theta,
Y)\big)/c_{\simulationepsilon,\epsilon}^2$.
The convergence $n S_{\simulationepsilon,\epsilon}^{(n)}(f) \to v_{\simulationepsilon,\epsilon}(f)$ 
follows from \citep[Theorem 9]{vihola-helske-franks}.
\end{proof} 

\begin{proof}[of Theorem \ref{thm:simple-acf}] 
The invariant distribution of \abcmcmc($\delta$) may be written as
$\tilde{\pi}_{\simulationepsilon}(\theta,y) = \pi_{\simulationepsilon}(\theta)
\bar{g}_{\simulationepsilon}(y\mid
\theta)$ where $\bar{g}_{\simulationepsilon}(y\mid \theta) \defeq 
\protect{g(y\mid \theta)}
\charfun{d(y,y^*)\le \simulationepsilon}
/L_{\simulationepsilon}(\theta)$, and 
that $\int \bar{g}_{\simulationepsilon}(y\mid\theta)
w_{\simulationepsilon,\epsilon}^p(y) 
\ud y = \bar{w}_{\simulationepsilon,\epsilon}(\theta)$ for $p\in\{1,2\}$. 
Consequently,
$\tilde{\pi}_{\simulationepsilon}(h_{\simulationepsilon,\epsilon}) = 
\pi_{\simulationepsilon}(f_{\simulationepsilon,\epsilon})$
and $\tilde{\pi}_{\simulationepsilon}(h_{\simulationepsilon,\epsilon}^2) =
\pi_{\simulationepsilon}(f^2 \bar{w}_{\simulationepsilon,\epsilon} )$, so 
$
    \var_{\tilde{\pi}_{\simulationepsilon}}(h_{\simulationepsilon,\epsilon}) 
    =  \var_{\pi_{\simulationepsilon}}(f_{\simulationepsilon,\epsilon}) + 
    \pi_{\simulationepsilon}\big( \bar{w}_{\simulationepsilon,\epsilon}
    (1-\bar{w}_{\simulationepsilon,\epsilon}) f^2
    \big).
$
Hereafter, let $a_{\simulationepsilon,\epsilon}\defeq
\big(\var_{\tilde{\pi}_{\simulationepsilon}}(h_{\simulationepsilon,\epsilon})\big)^{-1/2}$ and denote
$\tilde{h}_{\simulationepsilon,\epsilon} \defeq a_{\simulationepsilon,\epsilon}
h_{\simulationepsilon,\epsilon}$
and $\tilde{f}_{\simulationepsilon,\epsilon} \defeq a_{\simulationepsilon,\epsilon}
f_{\simulationepsilon,\epsilon}$. Clearly, 
$\var_{\tilde{\pi}_{\simulationepsilon}}(\tilde{h}_{\simulationepsilon,\epsilon})=1$ and
\[
\rho_k^{(\simulationepsilon,\epsilon)} = 
    e_k^{(\simulationepsilon,\epsilon)} -
    \big(\pi_{\simulationepsilon}(\tilde{f}_{\simulationepsilon,\epsilon})\big)^2,\qquad\text{where}\qquad
    e_k^{(\simulationepsilon,\epsilon)} \defeq
    \E \big[\tilde{h}_{\simulationepsilon,\epsilon}(\Theta_0^{(s)},Y_0^{(s)})
    \tilde{h}_{\simulationepsilon,\epsilon}(\Theta_k^{(s)}, Y_k^{(s)}) \big].
\]
Note that with $\phi=\phi_{\mathrm{simple}}$, the acceptance ratio
is 
$\alpha_{\simulationepsilon}(\theta,y;\hat{\theta},\hat{y})
    = \dot{\alpha}(\theta,\hat{\theta})
      \charfun{d(\hat{y},y^*)\le \simulationepsilon}$, where 
$\dot{\alpha}(\theta,\hat{\theta}) 
    = \min\big\{1, \mathrm{pr}(\hat{\theta})
        q(\hat{\theta},\theta) \big/\big(\mathrm{pr}(\theta)
        q(\theta,\hat{\theta})\big) \big\},
$
which is independent of $y$, so $(\Theta_k^{(s)})$ is
marginally a
Metropolis--Hastings type chain, with proposal $q$ and acceptance probability
$\alpha(\theta,\hat{\theta}) L_{\simulationepsilon}(\hat{\theta})$, and
\begin{align*}
&\E\big[\tilde{h}_{\simulationepsilon,\epsilon}(\Theta_1^{(s)}, Y_1^{(s)})
\bigmid (\Theta_0^{(s)},Y_0^{(s)})=(\theta,y)\big] - 
r_{\simulationepsilon}(\theta) \tilde{h}_{\simulationepsilon,\epsilon}(\theta,y) \\
&=a_{\simulationepsilon,\epsilon} \int q(\theta,\hat{\theta}) \dot{\alpha}(\theta,\hat{\theta}) 
    g(\hat{y}\mid \hat{\theta}) w_{\simulationepsilon,\epsilon}(\hat{y})
    f(\hat{\theta}) \ud \hat{\theta}\ud \hat{y}
    = \int q(\theta,\hat{\theta})\dot{\alpha}(\theta,\hat{\theta})
     L_{\simulationepsilon}(\hat{\theta}) 
	    \tilde{f}_{\simulationepsilon,\epsilon}(\hat{\theta})
    \ud \hat{\theta}.
\end{align*}
Using this iteratively, we obtain that
\begin{equation*}
    \textstyle
    e_k^{(\simulationepsilon,\epsilon)} = \E \big[
    \tilde{f}_{\simulationepsilon,\epsilon}(\Theta_0^{(s)})
    \tilde{f}_{\simulationepsilon,\epsilon}(\Theta_k^{(s)})
    \big] + 
    \int \tilde{\pi}_{\simulationepsilon}(\theta, y) \big[
    \tilde{h}_{\simulationepsilon,\epsilon}^2(\theta,y) -
    \tilde{f}_{\simulationepsilon,\epsilon}^2(\theta)\big]
    r_{\simulationepsilon}^k(\theta) \ud \theta \ud y,
\end{equation*}
and therefore with $\gamma_k^{(\simulationepsilon,\epsilon)} \defeq
a_{\simulationepsilon,\epsilon}^2 
\cov\big(f_{\simulationepsilon,\epsilon}(\Theta_0^{(s)}),
f_{\simulationepsilon,\epsilon}(\Theta_k^{(s)})\big)$,
\begin{equation*}
    \textstyle
    \sum_{k\ge 1} \rho_k^{(\simulationepsilon,\epsilon)}
    = \sum_{k\ge 1} \gamma_k^{(\simulationepsilon,\epsilon)} 
    + a_{\simulationepsilon,\epsilon}^2 \int \pi_{\simulationepsilon}(\theta) 
    \bar{w}_{\simulationepsilon,\epsilon}(\theta) \big( 1 -
    \bar{w}_{\simulationepsilon,\epsilon}(\theta)\big)
    r_{\simulationepsilon}(\theta)(1-r_{\simulationepsilon}(\theta))^{-1} f^2(\theta) \ud \theta.
\end{equation*}
We conclude by noticing that 
$2 \sum_{k\ge 1}
\gamma_k^{(\simulationepsilon,\epsilon)} = a_{\simulationepsilon,\epsilon}^2
\var_{\pi_{\simulationepsilon}}(f_{\simulationepsilon,\epsilon})
(\check{\tau}_{\simulationepsilon,\epsilon}(f) - 1)$.
\end{proof}


\setcounter{section}{1}
\def\appendixname{Supplement}

\appendix 
\section{Convergence of the tolerance adaptive ABC-MCMC under generalised conditions}
\label{app:ta-geometric} 

This section details a convergence theorem, under weaker assumptions than that of 
Theorem \mref{thm:ta-uniform}, for the tolerance adaptation 
(Algorithm \mref{alg:ta}) of Section \mref{sec:ta}.

For convenience, we
denote the distance distribution here as $T\sim Q_\theta(\uarg)$,
where $T\defeq d(Y,y^*)$ for $Y\sim g(\uarg|\theta)$.  
With this notation, and re-indexing $\Theta_k' = \tilde{\Theta}_{k+1}$, 
we may rewrite 
Algorithm \mref{alg:ta} as follows:
\begin{algo}
    \label{alg:ta-proofs}
Suppose $\Theta_0\in\mathsf{T}$ is a starting value with $\pr(\Theta_0)>0$.
Initialise $\simulationepsilon \defeq T_0 \sim Q_{\Theta_0}(\uarg)$. 
$k=0,\ldots, n_b-1$, iterate:
\begin{enumerate}[(i)]
\item Draw $\Theta_k' \sim q( \Theta_{k-1},\uarg)$ and
  $T_k' \sim Q_{\Theta_k'}(\uarg)$.
\item 
Accept, by setting 
  $(\Theta_{k+1}, T_{k+1}) \leftarrow (\Theta_k', T_k')$, with probability
  \begin{equation}\label{eq:ta-acceptance}
  \alpha_{\simulationepsilon_k}'(\Theta_k,T_k;\Theta_k', T_k') \defeq
  \min \bigg\{1, \frac{\pr(\Theta_k') q(\Theta_k',\Theta_k)
      \phi(T_k'/\simulationepsilon_k)}{\pr(\Theta_k)
      q(\Theta_k,\Theta_k')\phi(T_k/\simulationepsilon_k)}\bigg\} 
  \end{equation}
and otherwise reject, by setting $(\Theta_{k+1},T_{k+1}) \leftarrow (\Theta_k, T_k)$.
\item $\log \simulationepsilon_{k+1} \leftarrow \log
  \simulationepsilon_k + \gamma_{k+1} \big(\alpha^* -
  \alpha_{\simulationepsilon_k}'(\Theta_k,\Theta_k', T_k') \big)$.
\end{enumerate}
\end{algo}

Let us set $\beta\defeq \log \simulationepsilon$, and consider the
proposal-rejection Markov kernel 
  \begin{equation}\label{eq:marginal-kernel}
  \dot{P}_\beta(\theta, \ud \vartheta)
  \defeq
  q(\theta,\ud \vartheta)
  \alpha_\beta(\theta,\vartheta)
  + \bigg(1- \int q(\theta,\ud \vartheta ) \alpha_\beta(\theta,\vartheta)\bigg)\charfun{\theta\in \ud \vartheta},
  \end{equation}
  where
  $
  \alpha_\beta(\theta,\vartheta)
  \defeq
  \dot{\alpha}(\theta,\vartheta) L_\beta(\vartheta),
  $
  $$
  \dot{\alpha}(\theta,\vartheta)
  \defeq
  \min\bigg\{1,\frac{\pr(\vartheta) q(\vartheta,\theta)}{\pr(\theta) q(\theta,\vartheta)}\bigg\},
  \qquad
  \text{and}
  \qquad
L_\beta(\vartheta)
  \defeq
  \int Q_\vartheta(\ud t) \charfun{t\le e^\beta}.
  $$
  Then $\dot{P}_{\beta_k}$ is the transition of the
  $\theta$-coordinate chain of Algorithm \ref{alg:ta-proofs} with simple
  cut-off at iteration $k$, obtained by disregarding the
  $t$-coordinate.  It is easily seen to be reversible with respect to
  the posterior probability $ \pi_\beta(\theta) \propto \pr(\theta)
  L_\beta(\theta) $ given in \meqref{eq:abc-post}, written here in
  terms of $\beta=\log \simulationepsilon$ instead of
  $\simulationepsilon$.

\begin{assumption}
  \label{a:ta}
  Suppose $\phi= \phi_{\mathrm{simple}}$ and the following hold:
    \begin{enumerate}[(i)]
  \item\label{a:ta-step}
    Step sizes $(\gamma_k)_{k\ge 1}$ satisfy $\gamma_k \ge 0$, $\gamma_{k+1}\le \gamma_k$,
  $$
  \sum_{k\ge 1} \gamma_k = \infty,
  \qquad\text{and}\qquad
\sum_{k\ge 1} \gamma_k^2\Big(1 + \abs{\log \gamma_k} + \abs{\log \gamma_k}^2\Big) <\infty.
$$
\item The domain $\T\subset \R^{n_\theta}$, $n_\theta \ge 1$, is a nonempty open set.
    \item\label{a:ta-density}
      $\pr(\uarg)$ and  $q(\theta,\uarg)$ are uniformly bounded densities on $\R^{n_\theta}$ (i.e.~$\exists C>0$ s.t.~$q(\theta,\vartheta)<C$ and $\pr(\theta)<C$ for all $\theta,\,\vartheta\in\R^{n_\theta}$), and $\pr(\theta)=0$ for $\theta\notin \T$.  
  \item\label{a:ta-distance} 
    $Q_\theta(\ud t)$ admits a uniformly bounded density $Q_\theta(t)$.
\item\label{a:ta-compact}
  The values $\{\beta_k\}$ remain in some compact subset $\B\subset\R$ almost surely.
\item\label{a:ta-positive}
$c_\beta>0$ for all $\beta\in \B$, where $c_\beta\defeq \int \pr(\ud \theta) L_\beta(\theta)$.  
\item\label{a:ta-geometric}
There exists $\dot{V}:\T\to[1,\infty)$ such that the Markov transitions $\dot{P}_\beta$ are simultaneously $\dot{V}$-geometrically ergodic: there exist $C>0$ and $\rho\in (0,1)$ s.t.~for all $k\ge 1$ and $f:\T\to\R$ with $\abs{f}\le \dot{V}$, it holds that
$$
\abs{ \dot{P}_\beta^k f(\theta) - \pi_\beta(f) }\le C \dot{V}(\theta) \rho^k.
$$
\item\label{a:ta-variance}
With $\E[\uarg]= \E_{\theta,\beta}[\uarg]$ denoting expectation with
respect to the law of the marginal chain $(\Theta_k)$ of Algorithm
\ref{alg:ta-proofs} started at $\theta\in\T$, $\beta\in\B$, and with
$\dot{V}$ as in Assumption \ref{a:ta}\eqref{a:ta-geometric}, we have,
$$
\sup_{\theta,\beta, k} \E \big[ \dot{V}(\Theta_{k})^2\big] <\infty.
$$
\end{enumerate}
    \end{assumption}

    \begin{theorem}\label{thm:ta-geometric}
      Under Assumption \ref{a:ta}, the expected value of the
      acceptance probability \eqref{eq:ta-acceptance}, taken with
      respect to the stationary measure of the chain, converges to
      $\alpha^*$.
  \end{theorem}
Proof of Theorem \ref{thm:ta-geometric} can be found in Section
\ref{app:ta}.  It relies heavily on the simple conditions of
\citep[Theorem 2.3]{andrieu-moulines-priouret}, which says that one
must essentially show that the noise in the stochastic approximation
update is asymptotically controlled.

We remark that there are likely extensions of Assumption
\ref{a:ta}\eqref{a:ta-compact} to the general non-compact adaptation
parameter case based on projections
\citep[cf.][]{andrieu-moulines-priouret}.  


\section{Analysis of the tolerance adaptive ABC-MCMC}
\label{app:ta} 
      
In this section we aim to prove generalised convergence (Theorem
\ref{thm:ta-geometric} of Section \ref{app:ta-geometric}) of the
tolerance adaptation, from which Theorem \mref{thm:ta-uniform} of
Section \mref{sec:ta} will follow as a corollary. Throughout, we denote
by $C>0$ a constant which may change from line to line.

\subsection{Proposal augmentation}

Suppose $\dot{\genKernel}$ is a Markov kernel which can be written as
  \begin{equation}\label{eq:proposal-rejection}
\dot{\genKernel}(x, \ud y) = q(x,\ud y) \alpha(x,y) +\bigg(1-\int q(x,\ud y') \alpha(x,y')\bigg)\charfun{x\in \ud y},
  \end{equation}
where $\alpha(x,y)\in [0,1]$ is a jointly measurable function and $q(x,\ud y)$ is a Markov proposal kernel.
With $\bx\defeq (x,x')$, we define the \emph{proposal augmentation} to be the Markov kernel  
\begin{equation}\label{eq:proposal-augmentation}
\genKernel( \breve{x}, \ud \breve{y})
=
\alpha(\bx)\charfun{x'\in \ud y} q(x', \ud y') + \big(1-\alpha(\bx)\big) \charfun{x\in \ud y} q(x,\ud y').
\end{equation}
It is easy to see that $\genKernel$ need not be reversible even if $\dot{\genKernel}$ is reversible.  In this case, however, $\genKernel$ does leave a probability measure invariant.
\begin{lemma}\label{lem:proposal}
Suppose a Markov kernel $\dot{\genKernel}$ of the form given in \eqref{eq:proposal-rejection} is $\dot{\mu}$-reversible.  Let $\genKernel$ be its proposal augmentation.  Then the  following statements hold:
\begin{enumerate}[(i)]
\item \label{lem:proposal-measure}
  $\mu \genKernel = \mu$, where $\mu(\ud x,\ud x') \defeq \dot{\mu}(\ud x)q(x,\ud x')$.
\item\label{lem:proposal-ergodic}
  If $\dot{\genKernel}$ is $\dot{V}$-geometrically ergodic with constants $(\dot{C},\dot{\rho})$, then $\genKernel$ is $V$-geometrically ergodic with constants $(C, \rho)$, where
  $
  C\defeq 2 \dot{C}/\dot{\rho},
$
  $
  \rho \defeq \dot{\rho},
  $
  and
  $
V(\bx) \defeq \frac{1}{2}\big(V(x) + V(x')\big).
  $
  \end{enumerate}
\end{lemma}
Lemma \ref{lem:proposal} extends \citep[Theorem 4]{schuster-klebanov},
who consider the case where $\dot{P}$ is a Metropolis--Hastings chain
\citep[see also][]{delmas-jourdain,rudolf-sprungk18}.  The extension
to the more general class of reversible proposal-rejection chains
allows one to consider, for example, jump and delayed acceptance
chains, as well as the marginal chain \eqref{eq:marginal-kernel} of
Section \ref{app:ta-geometric}, which will be important for our
analysis of the tolerance adaptation.

\begin{proof}[of Lemma \ref{lem:proposal}]
  Part \eqref{lem:proposal-measure} follows by a direct calculation.
  We now consider part \eqref{lem:proposal-ergodic}.  
For $f:\X^2\to \R$, we shall use the notation
  $
\dot{f}(x) \defeq \int f(\bx) q(x,\ud x').
$
For $f:\X^2\to \R$, we have
$$
\int q(x, \ud x') \genKernel\big((x,x'); \ud \by\big) f(\by)
=
\int q(x,\ud x') \alpha(\bx) \dot{f}(x')
+ \int q(x,\ud x') \big(1 -\alpha(\bx)\big) \dot{f}(x)
= \dot{\genKernel}\dot{f}(x),
$$
and then inductively, for $k\ge 1$,
\begin{align*}
\int q(x, \ud x')  \genKernel^k\big( (x,x'); \ud \by\big) f(\by)
 &=
\int q(x,\ud x') \alpha(\bx) q(x',\ud y') L^{k-1}\big((x',y'); \ud \bz\big) f(\bz)
\\ &\qquad+\int q(x,\ud x')\big(1-\alpha(\bx)\big) q(x,\ud y') L^{k-1}\big( (x,y'); \ud \bz) f(\bz)
\\ &=
\int q(x,\ud x') \alpha(\bx) \dot{L}^{k-1} \dot{f}(x')
+\int q(x,\ud x')\big(1-\alpha(\bx)\big) \dot{L}^{k-1}\dot{f}(x)
\\&=
\dot{\genKernel}^k \dot{f}(x).
\end{align*}
We then have the equality,
\begin{align*}
\genKernel^k f(\bx)
&=
 \alpha(\bx)\int q(x', \ud y') \genKernel^{k-1}\big( (x', y'); \ud \bz\big) f(\bz)
+
\big(1 - \alpha(\bx)\big)\int  q(x,\ud y') \genKernel^{k-1}\big( (x,y'); \ud \bz\big) f(\bz)
\\ &=
\alpha(\bx)\dot{\genKernel}^{k-1} \dot{f}(x') + \big(1-\alpha(\bx)\big) \dot{\genKernel}^{k-1} \dot{f}(x).
\end{align*}
For $\norm{f}\le V$, note that $\lVert\dot{f}\rVert\le \dot{V}$ since $\norm{q}_\infty\le 1$, and we conclude  \eqref{lem:proposal-ergodic} from
\begin{align*}
  \abs{\genKernel^k f(\bx) - \mu(f)}
  &\le \alpha(\bx)\abs{\dot{\genKernel}^{k-1}\dot{f}(x') - \dot{\mu}(\dot{f})} + \big(1 - \alpha(\bx)\big) \abs{ \dot{\genKernel}^{k-1} \dot{f}(x) - \dot{\mu}(\dot{f})}
  \\ &\le
  \dot{C} \dot{\rho}^{k-1} \big( \dot{V}(x') + \dot{V}(x)\big).
\qedhere
\end{align*}
\end{proof}
Consider now the $\theta$-coordinate chain $\dot{P}_\beta$ presented
in \eqref{eq:marginal-kernel} of Section \ref{app:ta-geometric}. This
transition $\dot{P}_\beta$ is clearly a reversible proposal-rejection
chain of the form \eqref{eq:proposal-rejection}. We now consider
$P_\beta$, its proposal augmentation. This is the chain $\bTheta_k
\defeq (\Theta_k, \Theta_k') \in {\T}^2$, formed by disregarding the
$t$-parameter as with $\dot{P}_\beta$ before, but now augmenting by
the proposal $\theta'\sim q(\theta,\uarg)$.  Its transitions are of
the form $\btheta = \bTheta_{k}$ goes to $\bvartheta=\bTheta_{k+1}$ in
the ABC-MCMC, with $\bvartheta = (\vartheta,\vartheta')$ and kernel
$$
P_\beta(\btheta, \ud \bvartheta)
\defeq
\alpha_\beta(\btheta) \charfun{\theta'\in \ud \vartheta}  q(  \theta',\ud \vartheta')
+\big(1 - \alpha_\beta(\btheta)\big) \charfun{\theta\in \ud \vartheta } q( \theta,\ud \vartheta')
$$
By Lemma \ref{lem:proposal}\eqref{lem:proposal-measure}, $P_\beta$ leaves
$\pi_\beta' \defeq \pi_{\beta,u}'/c_\beta$ invariant, where
$
\pi_{\beta,u}'(\ud \btheta)
\defeq
\pr(\ud \theta) L_\beta(\theta) q( \theta,\ud \theta')
$
and
$
c_\beta \defeq \int \pr(\ud \theta) L_\beta(\theta).
$

\subsection{Monotonicity properties}
The following result establishes monotonicity of the mean field acceptance rate with increasing tolerance.
\begin{lemma}\label{lem:monotonicity}
  Assume Assumption \ref{a:ta}\eqref{a:ta-density} and \ref{a:ta}\eqref{a:ta-distance} hold.  The mapping
  $
  \beta
  \mapsto
  \pi_\beta'(\alpha_\beta)
  $
  is monotone non-decreasing.
  \end{lemma}
\begin{proof}
  Since $\pr(\theta)$ and $q(\theta,\theta')$ are uniformly bounded 
  (Assumption \ref{a:ta}\eqref{a:ta-density}), and $L_\beta(\theta)\le 1$, differentiation under the integral sign is possible in the following by the dominated convergence theorem.
By the quotient rule, 
\begin{equation}\label{eq:quotient}
\frac{\ud }{ \ud \beta}\Big( \pi_\beta'(\alpha_\beta)\Big)
=\frac{1}{c_\beta^2} \bigg( c_\beta \frac{\ud}{\ud \beta}\Big( \pi_{\beta,u}'(\alpha_\beta)\Big) - \pi_{\beta,u}'(\alpha_\beta) \frac{\ud c_\beta}{\ud \beta} \bigg).
\end{equation}
By reversibility of Metropolis--Hastings targeting $\pr(\theta)$ with proposal $q$, 
$$
\frac{\ud}{\ud \beta}\Big( \pi_{\beta,u}' (\alpha_\beta)\Big)
=
2 e^\beta \int \pr(\ud \theta) L_\beta(\theta) q(\theta, \ud \theta') \dot{\alpha}(\theta,\theta') Q_{\theta'}(e^\beta).
$$
With
$$
f(\theta')
\defeq
2 Q_{\theta'}(e^\beta) \int \pr(\ud \tilde{\theta}) L_\beta(\tilde{\theta}) - L_\beta(\theta')\int \pr(\ud \tilde{\theta}) Q_{\tilde{\theta}}(e^\beta),
$$
we can then write \eqref{eq:quotient} as  
$$
\frac{\ud}{\ud \beta}\Big( \pi_\beta'(\alpha_\beta)\Big)
=
\frac{e^\beta}{c_\beta^2} \int \pr(\ud \theta) L_\beta(\theta) q(\theta,\ud \theta') \dot{\alpha}(\theta,\theta')f(\theta').
$$
By the same reversibility property as before, we can write this again as
  $$
\frac{\ud}{\ud \beta}\Big( \pi_\beta'(\alpha_\beta)\Big)
=
\frac{e^\beta}{c_\beta^2} \int f(\theta) \pr(\ud \theta) \int q(\theta,\ud \theta')  L_\beta(\theta') \dot{\alpha}(\theta,\theta'), 
$$
We then conclude, since
$$
\int f(\theta) \pr(\ud \theta)
=\int Q_\theta(e^\beta) \pr(\ud \theta) \int L_\beta(\tilde{\theta}) \pr(\ud \tilde{\theta})\ge 0.
\qedhere
$$
\end{proof}

\begin{lemma}\label{lem:constant}
  The following statements hold:
  \begin{enumerate}[(i)]
\item \label{lem:constant-constant}
  The function $\beta \mapsto c_\beta$ is monotone non-decreasing on $\R$.
\item \label{lem:constant-compact}
  If Assumption \ref{a:ta}\eqref{a:ta-compact} and
  \ref{a:ta}\eqref{a:ta-positive} hold, then there exist $C_{\min}>0$,
  $C_{\max}>0$ such that $C_{\min} \le c_\beta \le C_{\max}$ for all
  $\beta\in\B$.  
\end{enumerate}
  \end{lemma}
\begin{proof}
Part \eqref{lem:constant-constant} follows,
for $\beta\le \beta'$, from
  $$
  c_\beta
  = \int \pr(\ud \theta) Q_\theta([0,e^{\beta}])
  \le \int \pr(\ud \theta) Q_\theta([0,e^{\beta'}])
  = c_{\beta'}.
  $$
  Consider part \eqref{lem:constant-compact}.  By part
  \eqref{lem:constant-constant} and compactness of $\B$ (Assumption
  \ref{a:ta}\eqref{a:ta-compact}), we can set $ C_{\min} \defeq
  c_{\min (\B) } $ and $ C_{\max} \defeq c_{\max (\B)}, $ both of
  which are positive by Assumption \ref{a:ta}\eqref{a:ta-positive}.
\end{proof}

\subsection{Stochastic approximation framework}\label{app:ta-sa}
To obtain a form common in the stochastic approximation literature
\citep[cf.][]{andrieu-moulines-priouret}, we write the update in
Algorithm \ref{alg:ta-proofs} as
\begin{align*}
\beta_{k+1}
&=
\beta_k + \gamma_{k+1} H_{\beta_k}(\bTheta_k, T_k')
\\ &=
 \beta_{k} + \gamma_{k+1} h(\beta_{k}) + \gamma_{k+1} \zeta_{k+1}
\end{align*}
where
$
H_{\beta}(\btheta, t')
\defeq
\alpha^* - \alpha_{\beta}'(\btheta, t'), 
$
$$
\alpha_{\beta}'(\btheta, t')
\defeq
\min \bigg\{1, \frac{ \pr(\theta') q(\theta',\theta)}{\pr(\theta) q(\theta,\theta')} \bigg\} \charfun{t' \le e^\beta},
$$
$$
h(\beta)
\defeq
\pi_\beta'(\widehat{H}_\beta)
=
\int \pi_\beta (\ud \theta) q( \theta,\ud \theta') Q_{\theta'}(\ud t') H_{\beta}(\theta,\theta', t'),
$$
noise sequence 
$
\zeta_{k+1}
\defeq
H_{\beta_k}(\bTheta_k, T_k') - h(\beta_k),
$
and conditional expectation
$$
\widehat{H}_\beta (\btheta)
\defeq
\E[ H_\beta(\bTheta,T')| \bTheta=\btheta],
$$
where $T'\sim Q_{\theta'}(\uarg)$.  We also set for convenience
$
\bar{H}_\beta(\btheta)\defeq
\widehat{H}_\beta(\btheta) - \pi_\beta'(\widehat{H}_\beta).
$

\begin{lemma}\label{lem:ergodic}
  Suppose Assumption \ref{a:ta}\eqref{a:ta-geometric} holds.  Then the following statements hold:
  \begin{enumerate}[(i)]
\item\label{lem:ergodic-simultaneous}
  The proposal augmented kernels $(P_\beta)_{\beta\in\B}$ are simultaneously $V$-geometrically ergodic, where $V(\theta,\theta')\defeq \frac{1}{2}\big(\dot{V}(\theta) + \dot{V}(\theta')\big)$, with $\dot{V}$ as in Assumption \ref{a:ta}\eqref{a:ta-geometric}.  
\item\label{lem:ergodic-poisson}
  There exists $C>0$, such that for all $\beta\in\B$, the formal solution  $g_\beta=\sum_{k\ge 0} P_\beta^k \bar{H}_\beta$ to the Poisson equation
  $
  g_\beta - P_\beta g_\beta = \bar{H}_\beta
  $
satisfies
  $
\abs{g_\beta(\btheta)} \le C V(\btheta).
  $
  \end{enumerate}
  \end{lemma}
\begin{proof}
  \eqref{lem:ergodic-simultaneous} follows directly from the explicit parametrisation for $(C,\rho)$ given in Lemma \ref{lem:proposal}\eqref{lem:proposal-ergodic}.

Part  \eqref{lem:ergodic-poisson} follows from part \eqref{lem:ergodic-simultaneous} and the bound, since $\abs{\bar{H}_\beta} \le 1\le V$,
$$
\abs{g_\beta(\btheta)}
\le
1+C_\beta \sum_{k\ge 1} \rho_\beta^k V(\btheta)
\le
\bigg(1+\frac{C_\beta}{1-\rho_\beta}\bigg) V(\btheta).
\qedhere
$$
\end{proof}

\subsection{Contractions}\label{app:ta-contractions}
We define for $V:\T \rightarrow [1,\infty)$ and $g:\T\to \R$ the $V$-norm
  $
\norm{g}_V \defeq \sup_{\theta\in \T} \frac{|g(\theta)|}{V(\theta)}.
$
We define for a bounded operator $A$ on a Banach space of bounded functions $f$, the operator norm $\norm{A}_\infty = \sup_f \frac{ \norm{Af}_\infty}{\norm{f}_\infty}$.

\begin{lemma}\label{lem:contractions}
Suppose Assumption \ref{a:ta}\eqref{a:ta-distance}, \ref{a:ta}\eqref{a:ta-compact} and \ref{a:ta}\eqref{a:ta-positive} hold.  The following hold:
  \begin{enumerate}[(i)]
  \item\label{lem:contractions-kernel}
    $\exists C>0$, $\exists C_{\B}^+>0$ s.t.~$\forall \beta_1 \in\B$, $\forall\beta_2 \in \B$, $\forall g:\T^2\to \R$ bounded, we have
    $$
     \norm{(P_{\beta_1} - P_{\beta_2})g}_\infty
    \le
    C \norm{g}_\infty \abs{e^{\beta_1} - e^{\beta_2} }
    \le
    C_{\B}^+ \norm{g}_\infty \abs{\beta_1 - \beta_2}.
    $$
  \item\label{lem:contractions-function} 
$\exists C_{\B}^->0$, $\exists C_{\B}>0$, s.t.~$\forall \beta_1\in \B$, $\forall\beta_2\in\B$, we have
    $$
    \norm{ \bar{H}_{\beta_1} - \bar{H}_{\beta_2} }_\infty
    \le
    C_{\B}^- \abs{e^{\beta_1}-e^{\beta_2} }
    \le
        C_{\B} \abs{ \beta_1 - \beta_2 }.
        $$
      \item
        \label{lem:contractions-measure} 
$\exists C_{\B}^->0$, $\exists C_{\B}>0$, s.t.~$\forall \beta_1\in \B$, $\forall\beta_2\in\B$, $\forall g:\T^2\to \R$ bounded, we have
    $$
    \abs{ \pi_{\beta_1}'(g) - \pi_{\beta_2}'(g) }
    \le
    C_{\B}^- \norm{g}_\infty \abs{e^{\beta_1}-e^{\beta_2} } 
    \le
        C_{\B} \norm{g}_\infty \abs{ \beta_1 - \beta_2 }.
        $$
      \end{enumerate}
  \end{lemma}
\begin{proof}
  By Assumption \ref{a:ta}\eqref{a:ta-distance}, we have for all $\beta_1,\, \beta_2\in \B$,
  $$
  \abs{L_{\beta_1}(\theta) - L_{\beta_2}(\theta)}
  =
  \int_{e^{\beta_1\wedge \beta_2}}^{e^{\beta_1\vee \beta_2}} Q_\theta(\ud t)
    \le
    C \abs{e^{\beta_1} - e^{\beta_2} }.
    $$
    We obtain the first inequality for part \eqref{lem:contractions-kernel}, then, from the bound,
    \begin{align*}
    \abs{(P_{\beta_1} - P_{\beta_2})g(\btheta)}
    &=\abs{
     \big( \alpha_{\beta_1}(\btheta) - \alpha_{\beta_2}(\btheta) \big) \dot{g}(\theta')
    + \big( \alpha_{\beta_2}(\btheta) - \alpha_{\beta_1}(\btheta) \big) \dot{g}(\theta)
    }
    \\ &\le
     \dot{\alpha}(\btheta) \abs{ L_{\beta_1}(\theta') - L_{\beta_2}(\theta')}
    \int \Big( q( \theta',\ud \vartheta') \abs{g(\theta', \vartheta')} + q(\theta,\ud \vartheta') \abs{g(\theta, \vartheta')}\Big),
    \end{align*}
The second, Lipschitz bound follows by a mean value theorem argument  for the function $\beta\mapsto e^\beta$, namely
 $$
 \abs{e^{\beta_1} - e^{\beta_2}}
 \le
 \sup_{\beta\in\B} e^{\beta}\, \abs{\beta_1 - \beta_2}
 \le C_{\B}^+\abs{\beta_1 - \beta_2},
 $$
 where the last inequality follows by compactness of $\B$ (Assumption \ref{a:ta}\eqref{a:ta-compact}).
 
 We now consider part \eqref{lem:contractions-function}. We have,
  $$
  \norm{\bar{H}_{\beta_1} - \bar{H}_{\beta_2}}_\infty
  \le 
  \lVert\widehat{H}_{\beta_1} - \widehat{H}_{\beta_2}\rVert_\infty + \abs{h(\beta_1) - h(\beta_2)}.
  $$
  For the first term, by Assumption \ref{a:ta}\eqref{a:ta-distance}, as in \eqref{lem:contractions-kernel}, we have
  $$
  \norm{\widehat{H}_{\beta_1} - \widehat{H}_{\beta_2}}_\infty
  \le \sup_{\btheta} 
  \dot{\alpha}(\btheta)\int_{e^{\beta_1\wedge \beta_2}}^{e^{\beta_1\vee\beta_2}} Q_{\theta'}(\ud t)
  \le C \abs{\beta_1 - \beta_2}.
  $$
  For the other term, we have
  $$
  \abs{h(\beta_1) - h(\beta_2)}
  \le
  \frac{1}{c_{\beta_1}}\abs{\pi_{\beta_1,u}'(\alpha_{\beta_1}) - \pi_{\beta_2,u}'(\alpha_{\beta_2})}
  +
  \pi_{\beta_2,u}'( \alpha_{\beta_2}) \frac{ \abs{ c_{\beta_1} - c_{\beta_2}} }{c_{\beta_1} c_{\beta_2}}.
  $$
  By the triangle inequality, we have
  $$
  \abs{\pi_{\beta_1,u}'(\alpha_{\beta_1}) - \pi_{\beta_2,u}'(\alpha_{\beta_2})}
  \le
  \abs{\pi_{\beta_1,u}'(\alpha_{\beta_1}) - \pi_{\beta_1,u}'(\alpha_{\beta_2})}
  +
  \abs{\pi_{\beta_1,u}'(\alpha_{\beta_2}) - \pi_{\beta_2,u}'(\alpha_{\beta_2})}
  $$
  Each term above is bounded by $C\abs{e^{\beta_1} - e^{\beta_2}}$, as is $\abs{c_{\beta_1} - c_{\beta_2}}$.  Moreover, by Lemma \ref{lem:constant}\eqref{lem:constant-compact}, we have $c_\beta\ge c_{\min}>0$ for all $\beta\in\B$, and the first inequality in part \eqref{lem:contractions-function} follows.
The second inequality follows by a mean value theorem argument as before.  Proof of \eqref{lem:contractions-measure} is simpler.
\end{proof}

\subsection{Control of noise}
We state a simple standard fact used repeatedly in the proof of Lemma \ref{lem:ta-key} below, our key lemma.
\begin{lemma}\label{lem:rv-monotone}
  Suppose $(X_j)_{j\ge 1}$ are random variables with
  $X_j \ge 0$, $X_{j+1}\le X_j$, and
  $
  \lim_{j\to\infty} \E[ X_j ] = 0.
  $
  Then, almost surely,
  $
  \lim_{j\to\infty} X_j = 0.
  $
  \end{lemma}

\begin{lemma}\label{lem:ta-key}
  Suppose Assumption \ref{a:ta} holds.  Then, with
  $
  \term_{j,n} \defeq \sum_{k=j}^n \gamma_k \zeta_k,
  $
  we have
  $$
  \lim_{j\to\infty} \sup_{n\ge j} \big|\term_{j,n} \big| =0,
  \qquad\text{almost surely.}
  $$
  \end{lemma}
\begin{proof}
Similar to \citep[][Proof of Prop. 5.2]{andrieu-moulines-priouret}, we write $\term_{j,n} \defeq \sum_{i=1}^8 \term_{j,n}^{(j)}$, where
  $$
  \widehat{H}_{\beta_{k-1}}(\bTheta_{k-1})
  =
\E[ H_{\beta_{k-1}} (\bTheta_{k-1}, T')| \mathcal{F}_{k-1}'],
  $$
with $\mathcal{F}_{k-1}'= \sigma(\beta_{k-1}, \Theta_{k-1},\Theta_{k-1}')$  representing the information obtained through running 
Algorithm \ref{alg:ta-proofs} up to and including iteration $k-2$ and then also generating $\Theta_{k-1}'$, and 
\allowdisplaybreaks
\begin{align*}
  \term_{j,n}^{(1)} &\defeq
  \sum_{k=j}^n \gamma_k\Big( H_{\beta_{k-1}}(\bTheta_{k-1},T_{k-1}') - \widehat{H}_{\beta_{k-1}}(\bTheta_{k-1}) \Big),
  \\
  \term_{j,n}^{(2)} &\defeq
\sum_{k=j}^n \gamma_k\Big( g_{\beta_{k-1}}(\bTheta_{k-1}) - P_{\beta_{k-1}} g_{\beta_{k-1}}(\bTheta_{k-2})   \Big),
\\
\term_{j,n}^{(3)} &\defeq
\gamma_{j-1}P_{j-1} g_{\beta_{j-1}}(\bTheta_{j-2}) - \gamma_n P_{\beta_{n}} g_{\beta_n}(\bTheta_{n-1}),
\\
\term_{j,n}^{(4)} &\defeq
\sum_{k=j}^{n} \Big( \gamma_k - \gamma_{k-1}\Big) P_{\beta_{k-1}} g_{\beta_{k-1}} (\bTheta_{k-2}),
  \label{eq:term-cuatro}
  \\
  \term_{j,n}^{(5)}&\defeq
  \sum_{k=j}^n \gamma_k \sum_{i\ge m_k+1} P_{\beta_{k}}^i \bar{H}_{\beta_k}(\bTheta_{k-1}),
  \\
  \term_{j,n}^{(6)}&\defeq
  -\sum_{k=j}^n \gamma_k \sum_{i\ge m_k+1} P_{\beta_{k-1}}^i \bar{H}_{\beta_{k-1}}(\bTheta_{k-1}),
  \\
  \term_{j,n}^{(7)}&\defeq
  \sum_{k=j}^n \gamma_k \sum_{i=1}^{m_k} \Big( P_{\beta_{k}}^i - P_{\beta_{k-1}}^i\Big) \bar{H}_{\beta_k}(\bTheta_{k-1}),
  \\
  \term_{j,n}^{(8)} &\defeq
  \sum_{k=j}^n \gamma_{k} \sum_{i=1}^{m_k} P_{\beta_{k-1}}^i\Big( \bar{H}_{\beta_{k}} - \bar{H}_{\beta_{k-1}}\Big) (\bTheta_{k-1}).
\end{align*}
Here, $g_\beta$ is the Poisson solution defined in Lemma \ref{lem:ergodic}\eqref{lem:ergodic-poisson},
and $m_k \defeq \lceil\abs{\log \gamma_k} \rceil$.
We remind that $\bar{H}_\beta \defeq \widehat{H}_\beta - h(\beta)$ from Section \ref{app:ta-sa}.

We now show
$
\lim_{j\to\infty} \sup_{n\ge j} \big|\term_{j,n}^{(i)} \big| =0
$
for each of the terms $i\in\{1{:}8\}$ individually, which implies the result of the lemma.

\noindent (1) Since for all $n>j$,
$$
\E[ \term_{j,n}^{(1)} - \term_{j,n-1}^{(1)}| \mathcal{F}_{n-1}']
=
0,
$$
we have that $(\term_{j,n}^{(1)})_{n\ge j}$ is a $\mathcal{F}_n'$-martingale for each $j\ge 1$.  By the Burkholder-Davis-Gundy inequality for martingales \citep[cf.][]{burkholder-davis-gundy}, we have
$$
\E[ \sup_{n\ge j} |\term_{j,n}^{(1)}|^2 ]
\le
C \E \Big[ \sum_{k=j}^\infty \gamma_k^2 \big( H_{\beta_{k-1}}(\bTheta_{k-1}, T_{k-1}') - \widehat{H}_{\beta_{k-1}}(\bTheta_{k-1})\big)^2 \Big]\le C \sum_{k=j}^\infty \gamma_k^2,
$$
where in the last inequality we have noted that $|H_\beta - \widehat{H}_\beta|\le 1$.  Since $\sum_{k\ge 1} \gamma_k^2 <\infty$, we get that
$$
\lim_{j\to\infty}
\E[ \sup_{n\ge j} |\term_{j,n}^{(1)}|^2 ]
=0.
$$
Hence, the result follows by Lemma \ref{lem:rv-monotone}.

\noindent(2) For $j\ge 2$, we have for $n>j$,
$$
\E[ \term_{j,n}^{(2)} - \term_{j,n-1}^{(2)}| \mathcal{F}_{n-2}'] =0,
$$
so that $(\term_{j,n}^{(2)})_{n\ge j}$ is a $\mathcal{F}_{n-1}'$-martingale, for $j\ge 2$.  By the Burkholder-Davis-Gundy inequality again,
$$
\E[ \sup_{n\ge j} \abs{\term_{j,n}^{(2)}}^2]
\le C \E\Big[ \sum_{k=j}^\infty \gamma_k^2\big( g_{\beta_{k-1}}(\bTheta_{k-1}) - P_{\beta_{k-1}} g_{\beta_{k-1}}(\bTheta_{k-2})\big)^2\Big].
$$
We then use Lemma \ref{lem:ergodic}\eqref{lem:ergodic-poisson} and $\norm{P_\beta}_\infty\le 1$, to get, after combining terms,
$$
\E[ \sup_{n\ge j} \abs{\term_{j,n}^{(2)}}^2]
\le
C \sum_{k=j-1}^\infty \gamma_k^2  \E\Big[V(\bTheta_{k-1})^2\Big]
\le
C \sum_{k=j-1}^\infty \gamma_k^2,
$$
where we have used Assumption \ref{a:ta}\eqref{a:ta-variance} in the last inequality.  We then conclude by Lemma \ref{lem:rv-monotone} as before.

\noindent (3)
By Lemma \ref{lem:ergodic}\eqref{lem:ergodic-poisson}, the triangle inequality, $\norm{P_\beta}_\infty\le 1$, and the dominated convergence theorem, we obtain
$$
\E[ \sup_{n\ge j} \abs{ \term_{j,n}^{(3)} }
  \le
  C \gamma_{j-1} \E[V(\bTheta_{j-2})] + C \sup_{n\ge j} \gamma_n \E[V(\bTheta_{n-1})].
$$
  We then apply Assumption \ref{a:ta}\eqref{a:ta-variance} and Jensen's inequality, and use that $\gamma_k$ go to zero, since $\sum \gamma_k^2 <\infty$, to get that
  $$
  \lim_{j\to\infty} \E[ \sup_{n\ge j} \abs{ \term_{j,n}^{(3)} }]
    \le
    C \Big(\lim_{j\to\infty} \gamma_{j-1} + \sup_{n\ge j} \gamma_n\Big) 
    =
    0.
  $$
We now may conclude by Lemma \ref{lem:rv-monotone}.
  
\noindent (4)
By Lemma \ref{lem:ergodic}\eqref{lem:ergodic-poisson} and $\gamma_{k}\le \gamma_{k-1}$, we have for $j\ge 2$, 
$$
\E[ \sup_{n\ge j} \abs{\term_{j,n}^{(4)}} ]
\le
C \sup_{n\ge j} \sum_{k=j}^n (\gamma_{k-1} - \gamma_{k}) \E[ V(\bTheta_{k-2})]
\le
C \sup_{n\ge j} \sum_{k=j}^n (\gamma_{k-1} - \gamma_{k})
$$
where we have used lastly Assumption \ref{a:ta}\eqref{a:ta-variance} and Jensen's inequality. Since this is a telescoping sum, we get
$$
\E[ \sup_{n\ge j} \abs{\term_{j,n}^{(4)}} ]
\le
C \sup_{n\ge j}  (\gamma_{j-1} - \gamma_{n}) 
\le
C \gamma_{j-1} 
$$
We then conclude by Lemma \ref{lem:rv-monotone}, since $\gamma_j\to 0$.  

\noindent (5)
By Lemma \ref{lem:ergodic}\eqref{lem:ergodic-simultaneous},
$
\abs{P_{\beta}^i \bar{H}_\beta(\btheta)} \le C\rho^i V(\btheta), 
$
where $C,\,\rho$ do not depend on $\beta\in \B$.  Hence, 
$$
\E[ \abs{\term_{j,n}^{(5)} } ]
\le
C \sum_{k=j}^n \gamma_k \sum_{i \ge m_k +1} \rho^i \E[V(\bTheta_{k-1}) ]
\le
C \sum_{k=j}^n \gamma_k \rho^{m_k},
$$
where we have used lastly Assumption \ref{a:ta}\eqref{a:ta-variance} and Jensen's inequality.  Since $m_k$ was defined to be of order $\abs{\log \gamma_k}$, we have 
$$
\E[ \abs{\term_{j,n}^{(5)} } ]
\le
C \sum_{k=j}^\infty \gamma_k^2 <\infty
$$
By the dominated convergence theorem, we then have 
$$
\E[ \sup_{n\ge j}\abs{\term_{j,n}^{(5)} } ]
\le
C \sum_{k=j}^\infty \gamma_k^2.
$$
Taking the limit $j\to\infty$, we can then conclude by using Lemma \ref{lem:rv-monotone}.

\noindent (6) The proof is essentially the same as for (5).  

\noindent (7)  We write for $i\ge 1$,
\begin{equation*}
P_{\beta_k}^i - P_{\beta_{k-1}}^i
=
\sum_{l=0}^{i-1} P_{\beta_k}^{i-l -1} \big(P_{\beta_k} - P_{\beta_{k-1}}\big) P_{\beta_{k-1}}^{l}.
\end{equation*}
Since $\norm{P_\beta^i}_\infty \le 1$ for all $i\ge 0$, and $\abs{\bar{H}_\beta} \le 1$, by Lemma \ref{lem:contractions}\eqref{lem:contractions-kernel}, we have
$$
\norm{ (P_{\beta_k}^i - P_{\beta_{k-1}}^i) \bar{H}_{\beta_k} }
\le
C \sum_{l=0}^{i-1}
\norm{ P_{\beta_k}^{i-l -1} }_\infty \abs{\beta_{k} - \beta_{k-1}} \norm{P_{\beta_{k-1}}^{l} \bar{H}_{\beta_k}}_\infty
\le
C  \abs{\beta_{k} - \beta_{k-1}} i.
$$
Since $\abs{\beta_{k} - \beta_{k-1}}\le \gamma_k$ from the adaptation
step in Algorithm \ref{alg:ta-proofs}, we have
$$
\abs{\term_{j,n}^{(7)}}
\le
C \sum_{k=j}^n \gamma_k \sum_{i=1}^{m_k} i \gamma_k
\le 
C \sum_{k=j}^\infty \gamma_k^2 m_k( 1+m_k)
<
\infty.
$$
We then take $\sup_{n\ge j}$ on the left, take the expectation, and conclude by Lemma \ref{lem:rv-monotone}.   

\noindent (8)
Since $\norm{P_\beta^i}_\infty \le 1$ and by Lemma \ref{lem:contractions}\eqref{lem:contractions-function}, we have that
$$
\norm{P_{\beta_{k-1}}^i( \bar{H}_{\beta_k} - \bar{H}_{\beta_{k-1}}) }_\infty
\le
\norm{P_{\beta_{k-1}}^i}_\infty \norm{ \bar{H}_{\beta_k} - \bar{H}_{\beta_{k-1}} }_\infty
\le
C \abs{\beta_k - \beta_{k-1}}
$$
Since $\abs{\beta_k - \beta_{k-1}} \le \gamma_k$, we have
$$
\E[ \sup_{n\ge j} \term_{j,n}^{(8)} ]
\le
C \sum_{k=j}^\infty \gamma_k^2 m_k <\infty.
$$
We then conclude by Lemma \ref{lem:rv-monotone}.
\end{proof}

\subsection{Proofs of convergence theorems}

\begin{proof}[of Theorem \ref{thm:ta-geometric}]
 We define our Lyapunov function $w:\R \rightarrow [0,\infty)$ to be the continuously differentiable function $w(\beta) \defeq \frac{1}{2}|e^\beta - e^{\beta^*}|^2$.
   We also have that $h(\beta)\defeq \pi_\beta'(\widehat{H}_\beta)$ is continuous, which follows from Lemma \ref{lem:contractions}\eqref{lem:contractions-measure}.
   One can then check that Assumption \ref{a:ta} and Lemma \ref{lem:ta-key} imply that the assumptions of \citep[Theorem 2.3]{andrieu-moulines-priouret} hold.  The latter result implies
   $\lim \abs{\beta_k - \beta^*} \rightarrow 0$, for some $\beta^*\in\B$ satisfying $\pi_{\beta^*}'(\alpha_{\beta^*})=\alpha^*$, as desired.
  \end{proof}

\begin{lemma}\label{lem:uniform}
Suppose Assumption \mref{a:uniform} holds.  Then both $(\dot{P}_\beta)_{\beta\in\B}$ and $(P_\beta)_{\beta\in\B}$ are simultaneously $1$-geometrically ergodic (i.e.~uniformly ergodic). 
  \end{lemma}
\begin{proof}
  We have $\pr(\theta) \le C_{\pr}$ some $C_{\pr}>0$, and also $0<\delta_q\le q(\theta,\vartheta)$, for all $\theta,\,\vartheta\in\T$.  Hence, for $A\subset \T$,
$$
  \dot{P}_\beta(\theta, A)
  \ge
  \int \delta_q \min\bigg\{ 1, \frac{\pr(\vartheta)}{\pr(\theta)} \bigg\} L_\beta(\vartheta) \charfun{\vartheta\in A}
  \ge
  \int \delta_q \frac{\pr(\vartheta)}{C_{\pr}} L_\beta(\vartheta) \charfun{\vartheta\in A}
  $$
  By Lemma \ref{lem:constant}\eqref{lem:constant-compact}, it holds $c_\beta\ge C_{\min}$ for some $C_{\min}>0$ for all $\beta\in \B$.  Therefore,
  $$
\dot{P}_\beta(\theta,A)\ge  \delta \pi_\beta(A),
  $$
  where $\delta_{\dot{P}} \defeq  \delta_q C_{\min}/C_{\pr} >0$ is independent of $\beta$.
As in Nummelin's split chain construction \citep[cf{.}][]{meyn-tweedie}, we can then define the Markov kernel 
$
R_\beta(\theta, A)
\defeq
(1-\delta_{\dot{P}})^{-1}\big( \dot{P}_\beta(\theta,A) -
\delta_{\dot{P}} \pi_\beta(A)\big)
$
with $\pi_\beta R_\beta = \pi_\beta$.  Set $\Pi_\beta(\theta,A) \defeq \pi_\beta(A)$.  
For any $f\le 1$, $\beta \in \B$, and $k\ge 1$, we have
\begin{align*}
  \norm{\dot{P}_\beta^k f - \pi_\beta(f)}_\infty
  &=
  (1-\delta_{\dot{P}})\norm{(R_\beta - \Pi_\beta) \dot{P}_\beta^{k-1} f}_\infty
  =
  (1-\delta_{\dot{P}}) \norm{ R_\beta \dot{P}_\beta^{k-1}\big(f - \pi_\beta(f)\big)}_\infty
  \\ &\le
  (1-\delta_{\dot{P}}) \norm{\dot{P}_\beta^{k-1}\big(f - \pi_\beta(f)\big)}_\infty
  =
  (1-\delta_{\dot{P}}) \norm{\dot{P}_\beta^{k-1}f - \pi_\beta(f)}_\infty
  \\& \le
  \ldots
  \le
  (1-\delta_{\dot{P}})^k \norm{ f- \pi_\beta(f)}_\infty
  \le
  2 (1-\delta_{\dot{P}})^k\norm{f}_\infty,
\end{align*}
where we have used $\norm{R_\beta}_\infty\le 1$ in the first inequality.
Hence, $(\dot{P}_\beta)_{\beta\in\B}$ are simultaneously $1$-geometrically ergodic, and thus so are $(P_\beta)_{\beta\in\B}$ by Lemma \ref{lem:ergodic}\eqref{lem:ergodic-simultaneous}.  
\end{proof}

\begin{proof}[of Theorem \mref{thm:ta-uniform}] 
  Since $(\dot{P}_\beta)_{\beta\in\B}$ are simultaneously
  $1$-geometric ergodic by Lemma \ref{lem:uniform}, it is direct to see that Assumption \mref{a:uniform} implies Assumption \ref{a:ta}.
    We conclude by Theorem \ref{thm:ta-geometric}.
\end{proof}

\section{Simultaneous tolerance and covariance adaptation}
\label{app:ta-am}

\begin{algo}[TA-AM($n_b, \alpha^*$)]
  \label{alg:ta-am}
Suppose $\Theta_0\in\mathsf{T}\subset \R^{n_\theta}$ is a starting value 
with $\pr(\Theta_0)>0$ and
$\Gamma_0=\mathbf{1}_{n_\theta\times n_\theta}$ is the identity matrix.
\begin{enumerate}[1.]
\item Initialise $\simulationepsilon \defeq T_0$ where $T_0 \sim Q_{\Theta_0}(\uarg)$ and $T_0>0$.  Set $\mu_0 \defeq \Theta_0$.
\item
For $k=0,\ldots, n_b-1$, iterate:
\begin{enumerate}[(i)]
\item Draw $\Theta_k' \sim N(\Theta_k, (2.38^2/n_\theta) \Gamma_k)$
\item Draw $T_k'\sim Q_{\Theta_k'}(\uarg)$.
\item Accept, by setting 
  $(\Theta_{k+1}, T_{k+1}) \leftarrow (\Theta_k', T_k')$, with probability 
  \[
  \alpha_{\simulationepsilon_k}(\Theta_k,T_k;\Theta_k', T_k') \defeq
  \min \bigg\{1, \frac{\pr(\Theta_k') 
      \phi(T_k'/\simulationepsilon_k)}{\pr(\Theta_k)
      \phi(T_k/\simulationepsilon_k)}\bigg\}.
  \]
  Otherwise reject, by setting $(\Theta_{k+1},T_{k+1}) \leftarrow (\Theta_k, T_k)$.
\item
  $
\log \simulationepsilon_{k+1} \leftarrow \log \simulationepsilon_k +
\gamma_{k+1} \big(\alpha^* -   \alpha_{\simulationepsilon_k}'(\Theta_k,\Theta_k', T_k') \big).
$
\item
  $
\mu_{k+1} \leftarrow \mu_k + \gamma_{k+1}\big(\Theta_{k+1} - \mu_k\big).
$
\item 
$
\Gamma_{k+1} \leftarrow \Gamma_k + \gamma_{k+1}\big( (\Theta_{k+1} -\mu_{k})(\Theta_{k+1} - \mu_k)^\transpose - \Gamma_k \big).
$
\end{enumerate}
\item Output $(\Theta_{n_b}, \simulationepsilon_{n_b})$.
\end{enumerate}
\end{algo}


\section{Details of extensions in Section \mref{sec:discussion}}
\label{app:extensions} 

In case of non-simple cut-off, the rejected
samples may be `recycled' the rejected samples in
the estimator \citep{ceperley-chester-kalos}. 
This may improve the accuracy (but can also reduce
accuracy in certain pathological cases; see \cite{delmas-jourdain}). 
The `waste recycling' estimator is 
\[
    E^\mathrm{WR}_{\simulationepsilon,\epsilon}(f) 
    \defeq \sum_{k=1}^n W_k^{(\simulationepsilon,\epsilon)}\big[
    \alpha_{\simulationepsilon}(\Theta_k,Y_k;\tilde{\Theta}_{k+1},\tilde{Y}_{k+1})f(\tilde{\Theta}_{k+1})
    + [1-\alpha_{\simulationepsilon}(\Theta_k, Y_k; \tilde{\Theta}_{k+1},\tilde{Y}_{k+1})]
    f(\Theta_k)\big].
\]
When $E_{\simulationepsilon,\epsilon}(f)$ is consistent under Theorem
\mref{thm:est-valid}, this is also a 
consistent estimator. Namely, as in the proof of Theorem
\mref{thm:est-valid}, we find that
$(\Theta_k,Y_k,\tilde{\Theta}_{k+1},Y_{k+1})_{k\ge 1}$ is a Harris
recurrent Markov chain with invariant distribution \[
\hat{\pi}_{\simulationepsilon}(\theta,y,\tilde{\theta},\tilde{y}) = 
\tilde{\pi}_{\simulationepsilon} (\theta, y) \tilde{q}(\theta, y;
\tilde{\theta},\tilde{y}), \] and
$\hat{\pi}_{\epsilon}(\theta,y,\tilde{\theta},\tilde{y})/
\hat{\pi}_{\simulationepsilon}(\theta,y,\tilde{\theta},\tilde{y}) = c_\epsilon
w_{\simulationepsilon,\epsilon}(y)$, where $\tilde{q}(\theta, y; \theta',y') =
q(\theta, \theta')g(y'\mid \theta')$. Therefore,
$E^\mathrm{WR}_{\simulationepsilon,\epsilon}(f)$ is a strongly consistent
estimator of \[
\E_{\hat{\pi}_{\epsilon}}\big[\alpha_{\simulationepsilon}(\Theta,Y;
\tilde{\Theta},\tilde{Y})f(\tilde{\Theta}) + [1-
\alpha_{\simulationepsilon}(\Theta,Y; \tilde{\Theta},\tilde{Y})]f(\Theta)\big]
= \E_{\pi_\epsilon}[f(\Theta)]. \] See
\citep[][]{rudolf-sprungk18,schuster-klebanov} for alternative waste
recycling estimators based on importance sampling analogues.  

A refined estimator may be formed as 
\[
    \hat{E}_{\simulationepsilon,\epsilon}(f) 
    = \textstyle \sum_{k=1}^n \sum_{j=0}^{m} \hat{U}_{k,j}^{(\simulationepsilon,\epsilon)}
    f(\Theta_k)
    \big/
    \sum_{\ell=1}^n \sum_{i=0}^{m} \hat{U}_{\ell,i}^{(\simulationepsilon,\epsilon)},
\]
where 
$\hat{U}_{k,0}^{(\simulationepsilon,\epsilon)} \defeq U_k^{(\simulationepsilon,\epsilon)}$
and $\hat{U}_{k,j}^{(\simulationepsilon,\epsilon)} \defeq
    \hat{N}_k\phi(\hat{T}_{k,j}/\epsilon)\big/\phi(T_k/\simulationepsilon)$,
for $j\ge 1$, and 
where $\hat{N}_k$ is the number of independent random variables
$\hat{Z}_1,\hat{Z}_2,\ldots\sim g(\uarg\mid \Theta_k)$ generated before
observing $\phi(\hat{T}_{k,\hat{N}_k}/\simulationepsilon)>0$.
The variables $\hat{T}_{k,j}\defeq d(\hat{Z}_j,y^*)$, and
$\hat{T}_k \defeq d(\hat{Y}_k,y^*)$ with independent
$\hat{Y}_k\sim g(\uarg\mid \Theta_k)$. 
This ensures that 
\[
    \E[\hat{N}_k\phi(\hat{T}_{k,j}/\epsilon)\mid \Theta_k=\theta, Y_k=y]
      = \frac{L_\epsilon(\theta)}{\P_{g(\uarg\mid
          \theta)}\big(\phi\big(d(Y,y^*)/\simulationepsilon\big)>0\big)},
      \]
which is sufficient to ensure that 
$\xi_{k,j}(f) \defeq \hat{U}_{k,j}^{(\simulationepsilon,\epsilon)} f(\Theta_k)$ is 
a proper weighting scheme from $\tilde{\pi}_{\simulationepsilon}$ to 
$\pi_\epsilon$; see 
\citep[Proposition 17(ii)]{vihola-helske-franks}, and
consequently the average
$\xi_k(f) \defeq (m+1)^{-1}\sum_{j=0}^m \xi_{k,j}(f)$ 
is a proper weighting.


\section{Supplementary results}
\label{app:results} 

\begin{figure}[H] 
\includegraphics[height=5cm]{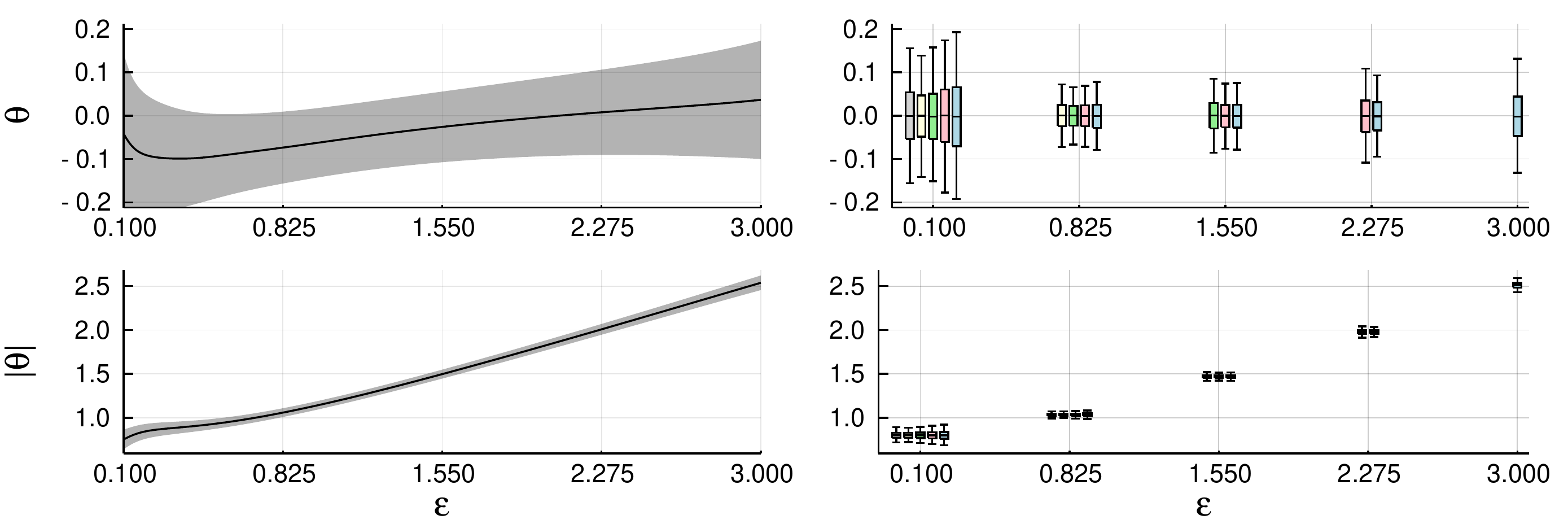}
\caption{Gaussian model with Gaussian cut-off.
  Estimates from single run of \abcmcmc($3$) (left) and
  estimates from 10,000 replications of
  \abcmcmc($\simulationepsilon$) for 
  $\simulationepsilon\in\{0.1, 0.825, 1.55, 2.275,
    3\}$ indicated by colours.} 
\label{fig:gauss-gauss}
\end{figure} 

\begin{figure}[H] 
\includegraphics[height=3.1875cm]{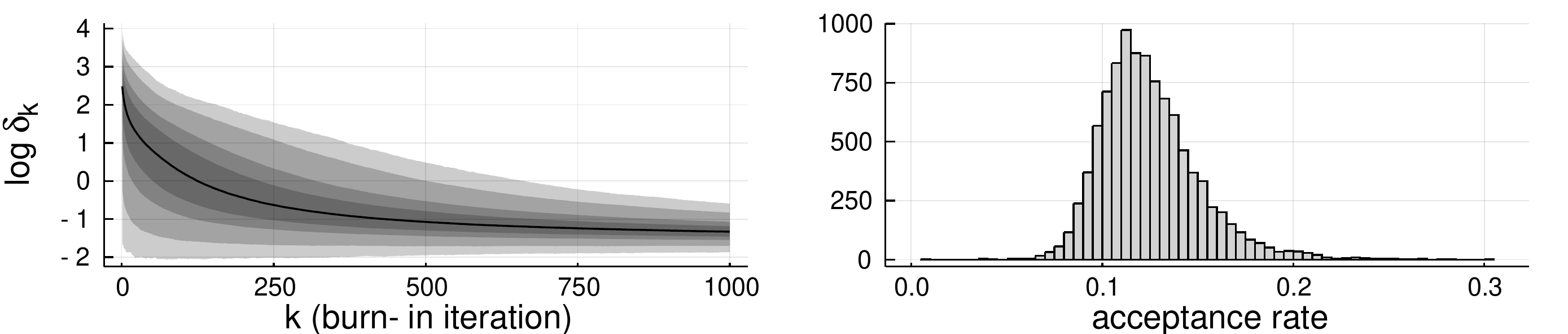}
\caption{Progress of tolerance adaptation (left) and histogram of acceptance
  rates (right) in the Gaussian model experiment with Gaussian
  cut-off.} 
\label{fig:NormalAdaptComparison-GaussCutoff}
\end{figure}

\begin{table}[H]
    \caption{Root mean square errors $(\times 10^{-2})$ from
  \abcmcmc($\simulationepsilon$) for tolerances $\epsilon$
  in the Gaussian mode with $\phi_{\mathrm{simple}}$.}
    \label{tab:gauss-rmse} 
    \small
    \begin{center}
    \begin{tabular}{lcccccccccccc}
        \toprule
        \small
        &&\multicolumn{5}{c}{$f(x)=x$}
        &\multicolumn{5}{c}{$f(x)=|x|$}
        & \raisebox{-1ex}{Acc.}\\
        \cmidrule(lr){3-7}
        \cmidrule(lr){8-12}
        Cut-off&\raisebox{-1.5pt}{$\simulationepsilon$} 
        $\backslash$
        \raisebox{3pt}{$\epsilon$}
        &0.10 & 0.82 & 1.55 & 2.28 & 3.00  
        &0.10 & 0.82 & 1.55 & 2.28 & 3.00  & rate\\
        \midrule
        \multirow{5}{*}{$\phi_{\mathrm{simple}}$}
& 0.1 & 9.68 &      &      &      &      & 5.54 &      &      &      &
   & 0.03\\
   & 0.82 & 8.99 & 3.81 &      &      &      & 5.38 & 2.14 &      &
&      & 0.22\\
& 1.55 & 9.21 & 3.66 & 3.59 &      &      & 5.5 & 2.17 & 1.96 &      &
   & 0.33\\
   & 2.28 & 9.67 & 3.86 & 3.6 & 3.97 &      & 5.85 & 2.28 & 2.02 &
   2.08 &      & 0.4\\
   & 3.0 & 10.36 & 4.03 & 3.71 & 3.98 & 4.51 & 6.21 & 2.42 & 2.12 &
   2.16 & 2.26 & 0.43\\
 \midrule
        \multirow{5}{*}{$\phi_{\mathrm{Gauss}}$} 
& 0.1 & 7.97 &      &      &      &      & 4.47 &      &      &      &
   & 0.05\\
   & 0.82 & 7.12 & 3.67 &      &      &      & 4.22 & 2.08 &      &
&      & 0.29\\
& 1.55 & 7.82 & 3.39 & 4.35 &      &      & 4.68 & 1.99 & 2.52 &
&      & 0.38\\
& 2.28 & 8.94 & 3.59 & 3.81 & 5.52 &      & 5.26 & 2.2 & 2.29 & 3.29 &
   & 0.41\\
   & 3.0 & 9.93 & 4.01 & 3.97 & 4.81 & 6.76 & 5.95 & 2.44 & 2.44 &
   2.92 & 4.1 & 0.42\\
\bottomrule
\end{tabular}
\end{center}
\end{table}


\begin{table}[H]
    \caption{Frequencies of the 95\% confidence intervals for the adaptive
      algorithm in the Gaussian model, for
      tolerance $\epsilon=0.1$.} 
      \small
\begin{center}
\begin{tabular}{lccccccccccccc} 
\toprule
& \multicolumn{6}{c}{$\phi_{\mathrm{simple}}$} 
& \multicolumn{6}{c}{$\phi_{\mathrm{Gauss}}$}  \\
    \cmidrule(lr){2-7} \cmidrule(lr){8-13} 
    & \multicolumn{5}{c}{Fixed tolerance} & Adapt  
    & \multicolumn{5}{c}{Fixed tolerance} & Adapt  \\
    \cmidrule(lr){2-6} \cmidrule(lr){7-7}
    \cmidrule(lr){8-12} \cmidrule(lr){13-13}
$\simulationepsilon$ & 0.1 & 0.82 & 1.55 & 2.28 & 3.0 & 0.64 & 0.1 & 0.82 & 1.55 & 2.28 &
3.0 & 0.28\\
\midrule
$x$& 0.93 & 0.97 & 0.97 & 0.98 & 0.98 & 0.96 & 0.93 & 0.94 & 0.94 & 0.95
& 0.95 & 0.93\\ 
$|x|$& 0.92 & 0.95 & 0.96 & 0.96 & 0.96 & 0.96 & 0.93 & 0.92 & 0.94 & 0.95
& 0.95 & 0.92\\
\bottomrule
\end{tabular}
\end{center}
\end{table}

\begin{table}[H]
    \setlength{\tabcolsep}{0.6ex}
    \caption{Root mean square errors and acceptance rates in the
      Lotka-Volterra experiment.}
    \label{tab:lotkavolterra-covadapt-rmse} 
    \small
    \begin{center}
    \begin{tabular}{llcccccccccccccccc}
        \toprule
        \small
        &
        &\multicolumn{5}{c}{$f(\theta)=\theta_1$}
        &\multicolumn{5}{c}{$f(\theta)=\theta_2$}
        &\multicolumn{5}{c}{$f(\theta)=\theta_3$}
        & \raisebox{-1ex}{Acc.}
        \\
        \cmidrule(lr){3-7}
        \cmidrule(lr){8-12}
        \cmidrule(lr){13-17}
& 
\raisebox{-1.5pt}{$\simulationepsilon$} 
        $\backslash$ \raisebox{3pt}{$\epsilon$}
& 80 & 110 & 140 & 170 & 200
& 80 & 110 & 140 & 170 & 200
& 80 & 110 & 140 & 170 & 200
& rate \\
\midrule
\multirow{5}{*}{\rotatebox{90}{$\phi_{\mathrm{simple}}$}}
& 80 & 2.37 &      &      &      &      & 1.32 &      &      &
&      & 2.94 &      &      &      &      & 0.05\\
& 110 & 1.81 & 1.48 &      &      &      & 0.99 & 0.86 &      &
&      & 2.26 & 1.88 &      &      &      & 0.07\\
& 140 & 1.75 & 1.41 & 1.22 &      &      & 0.93 & 0.77 & 0.68 &
&      & 2.11 & 1.69 & 1.4 &      &      & 0.1\\
& 170 & 1.83 & 1.35 & 1.14 & 1.05 &      & 0.96 & 0.75 & 0.64 & 0.6
&      & 2.14 & 1.65 & 1.33 & 1.15 &      & 0.14\\
& 200 & 1.93 & 1.41 & 1.11 & 0.97 & 0.95 & 1.06 & 0.75 & 0.61 & 0.56
& 0.6 & 2.37 & 1.74 & 1.36 & 1.16 & 1.09 & 0.17\\
\midrule 
\multirow{5}{*}{\rotatebox{90}{regr. $\phi_{\mathrm{Epa}}$}}
& 80 & 3.1 &      &      &      &      & 1.52 &      &      &      &
   & 2.77 &      &      &      &      & 0.05\\
   & 110 & 2.74 & 1.99 &      &      &      & 1.39 & 1.0 &      &
&      & 2.53 & 1.81 &      &      &      & 0.07\\
& 140 & 3.02 & 2.08 & 1.56 &      &      & 1.54 & 1.05 & 0.79 &
&      & 2.76 & 1.9 & 1.39 &      &      & 0.1\\
& 170 & 3.09 & 2.13 & 1.6 & 1.31 &      & 1.61 & 1.09 & 0.83 & 0.69
&      & 2.85 & 1.95 & 1.46 & 1.16 &      & 0.14\\
& 200 & 3.19 & 2.2 & 1.68 & 1.36 & 1.15 & 1.63 & 1.1 & 0.84 & 0.71 &
0.63 & 2.91 & 2.04 & 1.52 & 1.21 & 1.01 & 0.17\\
\bottomrule
\end{tabular}
\end{center}
\end{table}

\begin{figure}[H] 
\includegraphics[height=7.5cm]{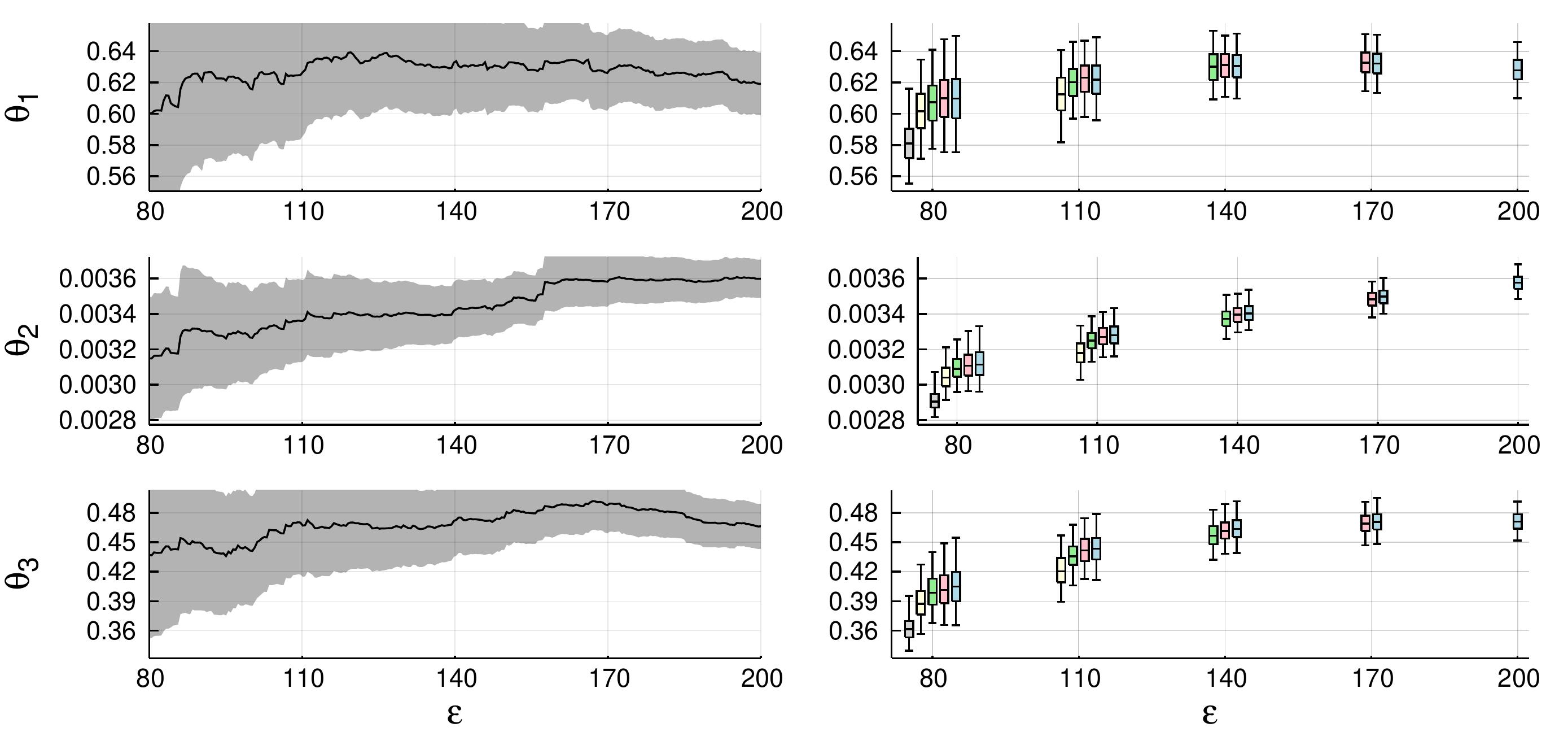}
\caption{Lotka-Volterra model with simple cut-off and step size
  $n^{-2/3}$.
Estimates from single run of \abcmcmc($200$) (left) and
  estimates from 1,000 replications of
  \abcmcmc($\simulationepsilon$) for 
  $\simulationepsilon\in\{80, 110, 140, 170,
    200\}$ indicated by colours.} 
\label{fig:lotkavolterra-stepsize-comparison}
\end{figure}

\begin{table}
    \caption{Coverages of confidence intervals 
      from \abcmcmc($\simulationepsilon$)
      for tolerance
      $\epsilon=80$, with fixed tolerance and with
      adaptive tolerance in the Lotka-Volterra model.} 
      \small
\begin{center}
 \begin{tabular}{lcccccccccccc} 
\toprule
    & \multicolumn{6}{c}{Post-correction, simple cut-off} 
    & \multicolumn{6}{c}{Regression, Epanechnikov cut-off} \\
    \cmidrule(lr){2-7} 
    \cmidrule(lr){8-13} 
    & \multicolumn{5}{c}{Fixed tolerance} & Adapt  
    & \multicolumn{5}{c}{Fixed tolerance} & Adapt  \\
    \cmidrule(lr){2-6} \cmidrule(lr){7-7}
    \cmidrule(lr){8-12} \cmidrule(lr){13-13}
$\simulationepsilon$ & 80.0 & 110.0 & 140.0 & 170.0 & 200.0 & 122.6 & 80.0 & 110.0 & 140.0
& 170.0 & 200.0 & 122.6\\
\midrule
$\theta_1$ & 0.8 & 0.97 & 0.99 & 0.99 & 1.0 & 0.93 & 0.75 & 0.92 & 0.93 & 0.93 &
0.96 & 0.9\\
$\theta_2$ & 0.73 & 0.94 & 0.98 & 0.98 & 0.99 & 0.84 & 0.76 & 0.93 & 0.94 & 0.96
& 0.98 & 0.9\\
$\theta_3$ & 0.74 & 0.94 & 0.98 & 0.99 & 0.99 & 0.86 & 0.68 & 0.87 & 0.9 & 0.92 &
0.95 & 0.83\\
\bottomrule
\end{tabular}
\end{center}
\end{table}


\begin{table}
    \setlength{\tabcolsep}{0.7ex}
    \caption{Frequencies of the 95\% confidence intervals and mean
      acceptance rates in the
      Lotka-Volterra experiment with step size $n^{-2/3}$.}
    \label{tab:lotkavolterra-stepsize-confidence} 
    \small
    \begin{center}
    \begin{tabular}{llcccccccccccccccc}
        \toprule
        \small
        &
        &\multicolumn{5}{c}{$f(\theta)=\theta_1$}
        &\multicolumn{5}{c}{$f(\theta)=\theta_2$}
        &\multicolumn{5}{c}{$f(\theta)=\theta_3$}
        & \raisebox{-1ex}{Acc.}
        \\
        \cmidrule(lr){3-7}
        \cmidrule(lr){8-12}
        \cmidrule(lr){13-17}
& 
\raisebox{-1.5pt}{$\simulationepsilon$} 
        $\backslash$ \raisebox{3pt}{$\epsilon$}
& 80 & 110 & 140 & 170 & 200
& 80 & 110 & 140 & 170 & 200
& 80 & 110 & 140 & 170 & 200
& rate \\
\midrule
\multirow{5}{*}{\rotatebox{90}{$\phi_{\mathrm{simple}}$}}
& 80 & 0.32 &      &      &      &      & 0.11 &      &      &
&      & 0.11 &      &      &      &      & 0.07\\
& 110 & 0.91 & 0.78 &      &      &      & 0.76 & 0.52 &      &
&      & 0.79 & 0.56 &      &      &      & 0.09\\
& 140 & 0.97 & 0.96 & 0.91 &      &      & 0.95 & 0.88 & 0.8 &
&      & 0.96 & 0.9 & 0.86 &      &      & 0.12\\
& 170 & 0.98 & 0.98 & 0.97 & 0.94 &      & 0.97 & 0.97 & 0.93 & 0.87
&      & 0.97 & 0.97 & 0.95 & 0.91 &      & 0.15\\
& 200 & 0.99 & 0.98 & 0.98 & 0.96 & 0.93 & 0.99 & 0.99 & 0.98 & 0.94
& 0.87 & 0.98 & 0.98 & 0.97 & 0.94 & 0.92 & 0.18\\
\midrule 
\multirow{5}{*}{\rotatebox{90}{regr. $\phi_{\mathrm{Epa}}$}}
& 80 & 0.34 &      &      &      &      & 0.41 &      &      &
&      & 0.25 &      &      &      &     & 0.07\\
& 110 & 0.86 & 0.81 &      &      &      & 0.88 & 0.87 &      &
&      & 0.81 & 0.82 &      &      &     & 0.09\\
& 140 & 0.94 & 0.94 & 0.92 &      &      & 0.95 & 0.95 & 0.96 &
&      & 0.9 & 0.91 & 0.91 &      &     & 0.12\\
& 170 & 0.95 & 0.95 & 0.95 & 0.95 &      & 0.96 & 0.97 & 0.97 & 0.97
&      & 0.92 & 0.94 & 0.94 & 0.95 &     & 0.15\\
& 200 & 0.95 & 0.95 & 0.95 & 0.95 & 0.94 & 0.96 & 0.97 & 0.97 & 0.98
& 0.98 & 0.92 & 0.94 & 0.95 & 0.95 & 0.95 & 0.18\\
\bottomrule
\end{tabular}
\end{center}
\end{table}

\begin{table}
    \setlength{\tabcolsep}{0.6ex}
    \caption{Root mean square errors and acceptance rates in the 
      Lotka-Volterra experiment with step size $n^{-2/3}$.}
    \label{tab:lotkavolterra-stepsize-rmse} 
    \small
    \begin{center}
    \begin{tabular}{llcccccccccccccccc}
        \toprule
        \small
        &
        &\multicolumn{5}{c}{$f(\theta)=\theta_1$}
        &\multicolumn{5}{c}{$f(\theta)=\theta_2$}
        &\multicolumn{5}{c}{$f(\theta)=\theta_3$}
        & \raisebox{-1ex}{Acc.}
        \\
        \cmidrule(lr){3-7}
        \cmidrule(lr){8-12}
        \cmidrule(lr){13-17}
& 
\raisebox{-1.5pt}{$\simulationepsilon$} 
        $\backslash$ \raisebox{3pt}{$\epsilon$}
& 80 & 110 & 140 & 170 & 200
& 80 & 110 & 140 & 170 & 200
& 80 & 110 & 140 & 170 & 200
& rate \\
\midrule
\multirow{5}{*}{\rotatebox{90}{$\phi_{\mathrm{simple}}$}}
& 80 & 3.24 &      &      &      &      & 2.2 &      &      &      &
   & 4.67 &      &      &      &      & 0.07\\
   & 110 & 2.12 & 2.14 &      &      &      & 1.14 & 1.38 &      &
 &      & 2.69 & 3.17 &      &      &      & 0.09\\
 & 140 & 1.87 & 1.56 & 1.49 &      &      & 0.89 & 0.81 & 0.79 &
 &      & 2.1 & 1.82 & 1.63 &      &      & 0.12\\
 & 170 & 1.77 & 1.27 & 1.05 & 0.96 &      & 0.87 & 0.68 & 0.59 &
 0.59 &      & 2.14 & 1.6 & 1.31 & 1.19 &      & 0.15\\
 & 200 & 1.94 & 1.45 & 1.2 & 1.11 & 1.08 & 0.95 & 0.69 & 0.59 & 0.54
 & 0.57 & 2.44 & 1.95 & 1.68 & 1.58 & 1.52 & 0.18\\
\midrule 
\multirow{5}{*}{\rotatebox{90}{regr. $\phi_{\mathrm{Epa}}$}}
 & 80 & 2.67 &      &      &      &      & 1.14 &      &      &
 &      & 2.17 &      &      &      &      & 0.07\\
 & 110 & 2.88 & 2.18 &      &      &      & 1.27 & 0.9 &      &
 &      & 2.36 & 1.76 &      &      &      & 0.09\\
 & 140 & 2.67 & 1.98 & 1.61 &      &      & 1.38 & 1.02 & 0.83 &
 &      & 2.57 & 1.91 & 1.54 &      &      & 0.12\\
 & 170 & 2.89 & 1.98 & 1.49 & 1.21 &      & 1.46 & 0.98 & 0.74 &
 0.61 &      & 2.63 & 1.79 & 1.34 & 1.08 &      & 0.15\\
 & 200 & 3.57 & 2.85 & 2.46 & 4.93 & 1.2 & 1.82 & 1.41 & 1.21 & 1.42
 & 0.63 & 3.11 & 2.32 & 1.88 & 1.81 & 1.22 & 0.18\\
 \bottomrule
\end{tabular}
\end{center}
\end{table}

\begin{table}
    \caption{Root mean square errors of estimators
      from \abcmcmc($\simulationepsilon$)
      for tolerance
      $\epsilon=80$, with fixed tolerance and with
      adaptive tolerance in the Lotka-Volterra model with step size
      $n^{-2/3}$.} 
    \label{tab:rmse-lotkavolterra-adaptation-vs-stepsize} 
\small
\begin{center}
\begin{tabular}{lcccccccccccc} 
\toprule
    & \multicolumn{6}{c}{Post-correction, simple cut-off} 
    & \multicolumn{6}{c}{Regression, Epanechnikov cut-off} \\
    \cmidrule(lr){2-7} 
    \cmidrule(lr){8-13} 
    & \multicolumn{5}{c}{Fixed tolerance} & Adapt  
    & \multicolumn{5}{c}{Fixed tolerance} & Adapt  \\
    \cmidrule(lr){2-6} \cmidrule(lr){7-7}
    \cmidrule(lr){8-12} \cmidrule(lr){13-13}
$\simulationepsilon$& 80 & 110 & 140 & 170 & 200 & 122.6 & 80 & 110 & 140
& 170 & 200 & 122.6\\
\midrule 
$\theta_1$ $(\times 10^{-2})$& 3.24 & 2.12 & 1.87 & 1.77 & 1.94 & 1.8 & 2.67 & 2.88 & 2.67 & 2.89 &
3.57 & 2.57\\
$\theta_2$ $(\times 10^{-4})$& 2.2 & 1.14 & 0.89 & 0.87 & 0.95 & 1.04 & 1.14 & 1.27 & 1.38 & 1.46 &
1.82 & 1.28\\
$\theta_3$ $(\times 10^{-2})$& 4.67 & 2.69 & 2.1 & 2.14 & 2.44 & 2.34 & 2.17 & 2.36 & 2.57 & 2.63 &
3.11 & 2.34\\
\bottomrule
\end{tabular}
\end{center}
\end{table}


\end{document}